\documentclass[journal]{IEEEtran}
\usepackage[utf8]{inputenc}
\usepackage{color}
\usepackage{array}
\usepackage{verbatim}
\usepackage{float}
\usepackage{amsmath}
\usepackage{amsthm}
\usepackage{amssymb}
\usepackage{graphicx}
\usepackage{longtable}
\usepackage{multirow}
\usepackage[unicode=true,
bookmarks=false,
breaklinks=false,pdfborder={0 0 1},colorlinks=false]
{hyperref}
\hypersetup{
	colorlinks,bookmarksopen,bookmarksnumbered,citecolor=blue,urlcolor=blue}
\usepackage{cite}

\floatstyle{ruled}
\newfloat{algorithm}{tbp}{loa}
\providecommand{\algorithmname}{Algorithm}
\floatname{algorithm}{\protect\algorithmname}

\makeatletter
\let\oldforeign@language\foreign@language
\DeclareRobustCommand{\foreign@language}[1]{%
	\lowercase{\oldforeign@language{#1}}}

\let\oldforeign@language\foreign@language
\DeclareRobustCommand{\foreign@language}[1]{%
	\lowercase{\oldforeign@language{#1}}}

\ifCLASSINFOpdf
\else
\fi

\newcommand{\MYfooter}{\smash{
		\hfil\parbox[t][\height][t]{\textwidth}{\centering
			\thepage}\hfil\hbox{}}}

\makeatletter

\def\ps@IEEEtitlepagestyle{%
	\def\@oddhead{\parbox[t][\height][t]{\textwidth}{\centering \scriptsize
			Personal use of this material is permitted. Permission from the author(s) and/or copyright holder(s), must be obtained for all other uses. Please contact us and provide details if you believe this document breaches copyrights.\\
			\noindent\makebox[\linewidth]{}
		}\hfil\hbox{}}%
	\def\@evenhead{\scriptsize\thepage \hfil \leftmark\mbox{}}%
	\def\@oddfoot{\parbox[t][\height][l]{\textwidth}{
			\vspace{-20pt}{\rule{\textwidth}{0.4pt}}\\ \footnotesize\underline{To cite this article:}
			{\bf{\textcolor{red}{H. A. Hashim and A. E. E. Eltoukhy "Landmark and IMU Data Fusion: Systematic Convergence Geometric Nonlinear Observer for SLAM and Velocity Bias," IEEE Transactions on Intelligent Transportation Systems, vol. 23, no. 4, pp. 3292-3301, 2022.}}} doi: \href{https://doi.org/10.1109/TITS.2020.3035550}{10.1109/TITS.2020.3035550}\\
			\noindent\makebox[\linewidth]
		}\hfil\hbox{}}%
	\def\@evenfoot{\MYfooter}}

\makeatother
\newtheorem{defn}{Definition}

\newtheorem{lem}{Lemma}

\newtheorem{prop}{Proposition}
\newtheorem{thm}{Theorem}
\newtheorem{rem}{Remark}

\newtheorem{assum}{Assumption}

\begin{document}
	\bstctlcite{IEEEexample:BSTcontrol}

	\title{Landmark and IMU Data Fusion: Systematic Convergence Geometric Nonlinear
	Observer for SLAM and Velocity Bias}

\author{Hashim A. Hashim$^*$\IEEEmembership{~Member, IEEE} and Abdelrahman E. E. Eltoukhy% <-this % stops a space
	\thanks{This work was supported in part by Thompson Rivers University Internal research fund \# 102315.}
	\thanks{$^*$Corresponding author, H. A. Hashim is with the Department of Engineering and Applied Science, Thompson Rivers University, Kamloops, British Columbia, Canada, V2C-0C8, e-mail: h.a.hashiim@gmail.com}
	\thanks{A. E. E. Eltoukhy is with the Department of Industrial and Systems Engineering, The Hong Kong Polytechnic University, Hung Hum, 
		Hong Kong e-mail: abdelrahman.eltoukhy@polyu.edu.hk}
}

\markboth{IEEE TRANSACTIONS ON INTELLIGENT TRANSPORTATION SYSTEMS, \today}{Hashim \MakeLowercase{\textit{et al.}}: Landmark and IMU Data Fusion: Systematic Convergence Geometric Nonlinear Observer for SLAM and Velocity Bias}

\markboth{}{Hashim \MakeLowercase{\textit{et al.}}: Landmark and IMU Data Fusion: Systematic Convergence Geometric Nonlinear Observer for SLAM and Velocity Bias}

\maketitle

\begin{abstract}
Navigation solutions suitable for cases when both autonomous robot's
pose (\textit{i.e}., attitude and position) and its environment are
unknown are in great demand. Simultaneous Localization and Mapping
(SLAM) fulfills this need by concurrently mapping the environment
and observing robot's pose with respect to the map. This work proposes
a nonlinear observer for SLAM posed on the manifold of the Lie group
of $\mathbb{SLAM}_{n}\left(3\right)$, characterized by systematic
convergence, and designed to mimic the nonlinear motion dynamics of
the true SLAM problem. The system error is constrained to start within
a known large set and decay systematically to settle within a known
small set. The proposed estimator is guaranteed to achieve predefined
transient and steady-state performance and eliminate the unknown bias
inevitably present in velocity measurements by directly using measurements
of angular and translational velocity, landmarks, and information
collected by an inertial measurement unit (IMU). Experimental results
obtained by testing the proposed solution on a real-world dataset
collected by a quadrotor demonstrate the observer's ability to estimate
the six-degrees-of-freedom (6 DoF) robot pose and to position unknown
landmarks in three-dimensional (3D) space.
\end{abstract}

% Note that keywords are not normally used for peerreview papers.

\IEEEpeerreviewmaketitle{}

\section{Introduction}

\IEEEPARstart{R}{obotic} mapping and localization are fundamental navigation tasks.
The nature of the problem is conditioned by the unknown component:
robot's pose, environment map or both. Simultaneous Localization and
Mapping (SLAM) consists in concurrent mapping of the environment and
identification of the robot's pose (\textit{i.e}, attitude and position).
SLAM solutions are indispensable in absence of absolute positioning
systems, such as global positioning systems (GPS), which are unsuitable
for occluded environments. SLAM problem is generally tackled using
a set of sensor measurements acquired at the body-fixed frame of the
robot in motion. Owing to the presence of uncertain elements in the
measurements, robust observers are an absolute necessity. 

SLAM has been an active area of research over the past thirty years
\cite{montemerlo2002fastslam,mullane2011random,durrant2006simultaneous,gong2019mapping,zlotnik2018SLAM,Hashim2020SLAMIEEELetter,lee2016ground,paz2008divide,holmes2008square,dong2019novel}.
Two traditional approaches to SLAM estimation are Gaussian filters
and nonlinear observers. Variations of Gaussian filters for SLAM designed
over the past decade include FastSLAM algorithm based on scalable
approach \cite{montemerlo2002fastslam}, MonoSLAM algorithm on single
camera \cite{holmes2008square}, iSAM that considers fast incremental
matrix factorization \cite{kaess2008isam}, extended Kalman filter
(EKF) with consistency analysis \cite{paz2008divide,bresson2015real},
invariant EKF \cite{zhang2017EKF_SLAM}, monocular visual-inertial
state estimator \cite{qin2018vins}, particle filter \cite{sim2005vision},
and linear time-variant Kalman filter for visual SLAM \cite{lourencco2016simultaneous}.
The feature that unites all of the above-mentioned algorithms is the
probabilistic framework. The three main challenges that characterize
the SLAM estimation problem are: 1) the added complexity of robot
motion in 3D space, 2) the duality of robot's pose and map estimation,
and above all 3) high nonlinearity of the problem. The true motion
dynamics of SLAM consist of robot's pose and landmark dynamics. Pose
dynamics of a robot are framed on the Lie group of the special Euclidean
group $\mathbb{SE}\left(3\right)$. Robot's orientation, also referred
to as attitude, is a fundamental part of landmark dynamics that belongs
to the Special Orthogonal Group $\mathbb{SO}\left(3\right)$. Due
to the fact that Gaussian filters fail to fully capture the true nonlinearity
of the SLAM estimation problem, nonlinear observers are deemed more
suitable.

Owing to the nonlinearity of the attitude and pose dynamics modeled
on $\mathbb{SO}\left(3\right)$ and $\mathbb{SE}\left(3\right)$,
respectively, over the past ten years several nonlinear observers
have been developed on $\mathbb{SO}\left(3\right)$ \cite{grip2012attitude,hashim2018SO3Stochastic,hashim2019SO3Wiley}
and $\mathbb{SE}\left(3\right)$ \cite{hashim2019SE3Det,zlotnik2018higher,hashim2020SE3Stochastic}.
Consequently, nonlinear observers developed on $\mathbb{SE}\left(3\right)$
have been found suitable in application to the SLAM problem \cite{strasdat2012local}.
The duality and nonlinearity aspects of the SLAM problem have been
tackled using a nonlinear observer for pose estimation and a Kalman
filter for feature estimation \cite{johansen2016globally}. However,
the work in \cite{johansen2016globally} did not fully capture the
true nonlinearity of the SLAM problem. To address this shortcoming,
a nonlinear observer for SLAM able to handle unknown bias attached
to velocity measurements has been proposed in \cite{zlotnik2018SLAM,Hashim2020SLAMIEEELetter}.
Despite significant progress, none of the above-mentioned observers
incorporated a measure that would guarantee error convergence for
the transient and steady-state performance. Systematic convergence
can be achieved and controlled by means of a prescribed performance
function (PPF) \cite{bechlioulis2008robust}. PPF entails trapping
the error to start within a predefined large set and reduce systematically
and smoothly to stay among predefined small set \cite{bechlioulis2008robust,hashim2017neuro,hashim2019SO3Wiley,hu2019lane}.
The error is constrained by dynamically reducing boundaries which
is achieved by employing its unconstrained form termed transformed
error. 

It must be emphasized that the SLAM problem is nonlinear and is posed
on the Lie group of $\mathbb{SLAM}_{n}\left(3\right)$. Taking into
consideration the shortcomings of the previously proposed observers
for SLAM, this work proposes a nonlinear observer that directly incorporates
measurements of angular and translational velocities as well as landmark
and IMU measurements distinguished by the following characteristics:
\begin{enumerate}
	\item[1)] The nonlinear observer for SLAM is developed directly on the Lie
	group of $\mathbb{SLAM}_{n}\left(3\right)$ mimicking the true SLAM
	dynamics with guaranteed measures of transient and steady-state performance.
	\item[2)] The error components of the Lyapunov function candidate have been
	proven to be asymptotically stable including the attitude error.
	\item[3)] The proposed observer successfully compensates for the unknown bias
	in group velocity measurements.
\end{enumerate}
Effectiveness and robustness of the proposed observer for six-degrees-of-freedom
(6 DoF) robot pose estimation and mapping of the unknown landmarks
in three-dimensional (3D) space have been confirmed experimentally
using a real-world dataset collected by an unmanned aerial vehicle.

This paper consists of six sections: Section \ref{sec:Preliminaries-and-Math}
provides a brief overview of mathematical notation and preliminaries,
and defines the Lie group of $\mathbb{SO}\left(3\right)$, $\mathbb{SE}\left(3\right)$,
and $\mathbb{SLAM}_{n}\left(3\right)$. Section \ref{sec:SE3_Problem-Formulation}
presents the true SLAM problem, the set of available measurements,
and error criteria. Section \ref{sec:SLAM_Filter} presents the concept
of PPF, mapping of the error to transformed error, and proposes the
nonlinear observer for SLAM on $\mathbb{SLAM}_{n}\left(3\right)$.
Section \ref{sec:SE3_Simulations} presents the obtained results.
Section \ref{sec:SE3_Conclusion} summarizes the work.

\section{Preliminaries of $\mathbb{SLAM}_{n}\left(3\right)$ \label{sec:Preliminaries-and-Math}}

The sets of real numbers, nonnegative real numbers, $n$-dimensional
Euclidean space, and $n$-by-$m$ dimensional space are denoted by
$\mathbb{R}$, $\mathbb{R}_{+}$, $\mathbb{R}^{n}$, and $\mathbb{R}^{n\times m}$,
respectively. $\left\Vert x\right\Vert =\sqrt{x^{\top}x}$ is an Euclidean
norm for $x\in\mathbb{R}^{n}$. $\mathbf{I}_{n}$ and $\underline{\mathbf{0}}_{n}$
represent an $n$-dimensional identity matrix and a zero column vector,
respectively. $\left\{ \mathcal{I}\right\} $ represents an inertial-fixed-frame
of reference, while $\left\{ \mathcal{B}\right\} $ represents a body-fixed-frame.
The Special Orthogonal Group $\mathbb{SO}\left(3\right)$ is defined
by
\[
\mathbb{SO}\left(3\right)=\left\{ \left.R\in\mathbb{R}^{3\times3}\right|RR^{\top}=R^{\top}R=\mathbf{I}_{3}\text{, }{\rm det}\left(R\right)=+1\right\} 
\]
where ${\rm det\left(\cdot\right)}$ refers to a determinant. Special
Euclidean Group $\mathbb{SE}\left(3\right)$ is described by
\[
\mathbb{SE}\left(3\right)=\left\{ \left.\boldsymbol{T}=\left[\begin{array}{cc}
R & P\\
\underline{\mathbf{0}}_{3}^{\top} & 1
\end{array}\right]\in\mathbb{R}^{4\times4}\right|R\in\mathbb{SO}\left(3\right),P\in\mathbb{R}^{3}\right\} 
\]
where $P\in\mathbb{R}^{3}$ refers to position and $R\in\mathbb{SO}\left(3\right)$
refers to orientation, commonly termed attitude, which together constitute
a homogeneous transformation matrix \cite{hashim2019SE3Det,hashim2020SE3Stochastic}
\begin{equation}
\boldsymbol{T}=\mathcal{Z}(R,P)=\left[\begin{array}{cc}
R & P\\
\underline{\mathbf{0}}_{3}^{\top} & 1
\end{array}\right]\in\mathbb{SE}\left(3\right)\label{eq:T_SLAM}
\end{equation}
that comprehensively describes the pose of a rigid-body in 3D space,
with $\underline{\mathbf{0}}_{3}$ referring to a zero column vector.
The Lie-algebra of $\mathbb{SO}\left(3\right)$ is
\[
\mathfrak{so}\left(3\right)=\left\{ \left.\left[\Omega\right]_{\times}\in\mathbb{R}^{3\times3}\right|\left[\Omega\right]_{\times}^{\top}=-\left[\Omega\right]_{\times},\Omega\in\mathbb{R}^{3}\right\} 
\]
where $\left[\Omega\right]_{\times}$ is a skew symmetric matrix and
the relevant map $\left[\cdot\right]_{\times}:\mathbb{R}^{3}\rightarrow\mathfrak{so}\left(3\right)$
is described by
\[
\left[\Omega\right]_{\times}=\left[\begin{array}{ccc}
0 & -\Omega_{3} & \Omega_{2}\\
\Omega_{3} & 0 & -\Omega_{1}\\
-\Omega_{2} & \Omega_{1} & 0
\end{array}\right]\in\mathfrak{so}\left(3\right),\hspace{1em}\Omega=\left[\begin{array}{c}
\Omega_{1}\\
\Omega_{2}\\
\Omega_{3}
\end{array}\right]
\]
For $\Omega,V\in\mathbb{R}^{3}$ one has $\left[\Omega\right]_{\times}V=\Omega\times V$
where $\times$ signifies a cross product. $\mathfrak{se}\left(3\right)$
is the Lie-algebra of $\mathbb{SE}\left(3\right)$ described by{\small{}
	\[
	\mathfrak{se}\left(3\right)=\left\{ \left[U\right]_{\wedge}\in\mathbb{R}^{4\times4}\left|\exists\Omega,V\in\mathbb{R}^{3}:\left[U\right]_{\wedge}=\left[\begin{array}{cc}
	\left[\Omega\right]_{\times} & V\\
	\underline{\mathbf{0}}_{3}^{\top} & 0
	\end{array}\right]\right.\right\} 
	\]
}where $\left[\cdot\right]_{\wedge}$ refers to a wedge operator and
the map $\left[\cdot\right]_{\wedge}:\mathbb{R}^{6}\rightarrow\mathfrak{se}\left(3\right)$
is described by
\[
\left[U\right]_{\wedge}=\left[\begin{array}{cc}
\left[\Omega\right]_{\times} & V\\
\underline{\mathbf{0}}_{3}^{\top} & 0
\end{array}\right]\in\mathfrak{se}\left(3\right),\hspace{1em}U=\left[\begin{array}{c}
\Omega\\
V
\end{array}\right]\in\mathbb{R}^{6}
\]
$\mathbf{vex}:\mathfrak{so}\left(3\right)\rightarrow\mathbb{R}^{3}$
describes the inverse mapping of $\left[\cdot\right]_{\times}$ such
that
\begin{equation}
\begin{cases}
\mathbf{vex}\left(\left[\Omega\right]_{\times}\right) & =\Omega,\hspace{1em}\forall\Omega\in\mathbb{R}^{3}\\
\left[\mathbf{vex}(\left[\Omega\right]_{\times})\right]_{\times} & =\left[\Omega\right]_{\times}\in\mathfrak{so}\left(3\right)
\end{cases}\label{eq:SLAM_VEX}
\end{equation}
$\boldsymbol{\mathcal{P}}_{a}$ denotes the anti-symmetric projection
$\boldsymbol{\mathcal{P}}_{a}:\mathbb{R}^{3\times3}\rightarrow\mathfrak{so}\left(3\right)$
where
\begin{equation}
\boldsymbol{\mathcal{P}}_{a}\left(X\right)=\frac{1}{2}(X-X^{\top})\in\mathfrak{so}\left(3\right),\hspace{1em}\forall X\in\mathbb{R}^{3\times3}\label{eq:SLAM_Pa}
\end{equation}
$\boldsymbol{\Upsilon}\left(\cdot\right)$ describes the composition
mapping $\boldsymbol{\Upsilon}=\mathbf{vex}\circ\boldsymbol{\mathcal{P}}_{a}$
as
\begin{equation}
\boldsymbol{\Upsilon}\left(X\right)=\mathbf{vex}\left(\boldsymbol{\mathcal{P}}_{a}\left(X\right)\right)\in\mathbb{R}^{3},\hspace{1em}\forall X\in\mathbb{R}^{3\times3}\label{eq:SLAM_VEX_a}
\end{equation}
For $R\in\mathbb{SO}\left(3\right)$, $\left\Vert R\right\Vert _{{\rm I}}$
describes the Euclidean distance as
\begin{equation}
\left\Vert R\right\Vert _{{\rm I}}=\frac{1}{4}{\rm Tr}\{\mathbf{I}_{3}-R\}\in\left[0,1\right]\label{eq:SLAM_Ecul_Dist}
\end{equation}
$\overset{\circ}{\mathcal{M}}$ and $\overline{\mathcal{M}}$ refer
to submanifolds of $\mathbb{R}^{4}$ described by
\begin{align*}
\overset{\circ}{\mathcal{M}} & =\left\{ \left.\overset{\circ}{x}=\left[\begin{array}{cc}
x^{\top} & 0\end{array}\right]^{\top}\in\mathbb{R}^{4}\right|x\in\mathbb{R}^{3}\right\} \\
\overline{\mathcal{M}} & =\left\{ \left.\overline{x}=\left[\begin{array}{cc}
x^{\top} & 1\end{array}\right]^{\top}\in\mathbb{R}^{4}\right|x\in\mathbb{R}^{3}\right\} 
\end{align*}
Let $\mathbb{SLAM}_{n}\left(3\right)=\mathbb{SE}\left(3\right)\times\overline{\mathcal{M}}^{n}$
be a Lie group described by
\begin{align*}
\mathbb{SLAM}_{n}\left(3\right) & =\left\{ X=(\boldsymbol{T},\overline{{\rm p}})\left|\boldsymbol{T}\in\mathbb{SE}\left(3\right),\overline{{\rm p}}\in\overline{\mathcal{M}}^{n}\right.\right\} \\
\overline{{\rm p}} & =[\overline{{\rm p}}_{1},\overline{{\rm p}}_{2},\ldots,\overline{{\rm p}}_{n}]\in\overline{\mathcal{M}}^{n}\\
\overline{\mathcal{M}}^{n} & =\overline{\mathcal{M}}\times\overline{\mathcal{M}}\times\cdots\times\overline{\mathcal{M}}
\end{align*}
where $\overline{{\rm p}}_{i}=\left[\begin{array}{cc}
{\rm p}_{i}^{\top} & 1\end{array}\right]^{\top}\in\overline{\mathcal{M}}$, for $i=1,2,\ldots,n$. $\mathfrak{slam}_{n}\left(3\right)=\mathfrak{se}\left(3\right)\times\overset{\circ}{\mathcal{M}}^{n}$
is the Lie algebra of $\mathbb{SLAM}_{n}\left(3\right)$, and the
tangent space of the identity element is described by
\begin{align*}
\mathfrak{slam}_{n}\left(3\right) & =\left\{ \mathcal{Y}=(\left[U\right]_{\wedge},\overset{\circ}{{\rm v}})\left|\left[U\right]_{\wedge}\in\mathfrak{se}\left(3\right),\overset{\circ}{{\rm v}}\in\overset{\circ}{\mathcal{M}}^{n}\right.\right\} \\
\overset{\circ}{{\rm v}} & =[\overset{\circ}{{\rm v}}_{1},\overset{\circ}{{\rm v}}_{2},\ldots,\overset{\circ}{{\rm v}}_{n}]\in\overset{\circ}{\mathcal{M}}^{n}\\
\overset{\circ}{\mathcal{M}}^{n} & =\overset{\circ}{\mathcal{M}}\times\overset{\circ}{\mathcal{M}}\times\cdots\times\overset{\circ}{\mathcal{M}}
\end{align*}
where $\overset{\circ}{{\rm v}}_{i}=\left[\begin{array}{cc}
{\rm v}_{i}^{\top} & 0\end{array}\right]^{\top}\in\overset{\circ}{\mathcal{M}}$, for $i=1,2,\ldots,n$. The identities below will be used in the
forthcoming derivations 
\begin{align}
\left[a\times b\right]_{\times}= & ba^{\top}-ab^{\top},\hspace{1em}b,a\in{\rm \mathbb{R}}^{3}\label{eq:SLAM_Identity2}\\
\left[Ra\right]_{\times}= & R\left[b\right]_{\times}R^{\top},\hspace{1em}a\in{\rm \mathbb{R}}^{3},R\in\mathbb{SO}\left(3\right)\label{eq:SLAM_Identity1}\\
{\rm Tr}\left\{ M\left[a\right]_{\times}\right\} = & {\rm Tr}\left\{ \boldsymbol{\mathcal{P}}_{a}\left(M\right)\left[a\right]_{\times}\right\} ,\hspace{1em}a\in{\rm \mathbb{R}}^{3},M\in\mathbb{R}^{3\times3}\nonumber \\
= & -2\mathbf{vex}\left(\boldsymbol{\mathcal{P}}_{a}\left(M\right)\right)^{\top}a\label{eq:SLAM_Identity6}
\end{align}

\begin{comment}
\begin{align}
\left[a\right]_{\times}^{2}= & -||a||^{2}\mathbf{I}_{3}+aa^{\top},\hspace{1em}a\in{\rm \mathbb{R}}^{3}\label{eq:SLAM_Identity4}\\
\left[a\times b\right]_{\times}= & ba^{\top}-ab^{\top},\hspace{1em}b,a\in{\rm \mathbb{R}}^{3}\label{eq:SLAM_Identity2-1}\\
M\left[a\right]_{\times} & +\left[a\right]_{\times}M={\rm Tr}\left\{ M\right\} \left[a\right]_{\times}-\left[Ma\right]_{\times}\nonumber \\
& ,\hspace{1em}a\in{\rm \mathbb{R}}^{3},M\in\mathbb{R}^{3\times3}\label{eq:SLAM_Identity5-1}\\
\left[Ra\right]_{\times}= & R\left[b\right]_{\times}R^{\top},\hspace{1em}a\in{\rm \mathbb{R}}^{3},R\in\mathbb{SO}\left(3\right)\label{eq:SLAM_Identity1-1}\\
{\rm Tr}\left\{ \left[a\right]_{\times}M\right\} = & 0,\hspace{1em}a\in{\rm \mathbb{R}}^{3},M=M^{\top}\in\mathbb{R}^{3\times3}\label{eq:SLAM_Identity3-1}\\
{\rm Tr}\left\{ M\left[a\right]_{\times}\right\} = & {\rm Tr}\left\{ \boldsymbol{\mathcal{P}}_{a}\left(M\right)\left[a\right]_{\times}\right\} ,\hspace{1em}a\in{\rm \mathbb{R}}^{3},M\in\mathbb{R}^{3\times3}\nonumber \\
= & -2\mathbf{vex}\left(\boldsymbol{\mathcal{P}}_{a}\left(M\right)\right)^{\top}a\label{eq:SLAM_Identity6-1}
\end{align}
\end{comment}

\section{SLAM Problem\label{sec:SE3_Problem-Formulation}}

\subsection{Available Measurements}

Suppose that the map contains $n$ landmarks as illustrated in Fig.
\ref{fig:SLAM}. Let $R\in\mathbb{SO}\left(3\right)$, $P\in\mathbb{R}^{3}$,
and ${\rm p}_{i}\in\mathbb{R}^{3}$ be rigid-body's orientation, rigid-body's
translation, and the $i$th landmark location in 3D space, respectively,
where $R\in\left\{ \mathcal{B}\right\} $, $P\in\left\{ \mathcal{I}\right\} $,
and ${\rm p}_{i}\in\left\{ \mathcal{I}\right\} $ for all $i=1,2,\ldots,n$.
SLAM problem considers a completely unknown environment, and therefore,
involves simultaneous estimation of 1) robot's pose $\boldsymbol{T}\in\mathbb{SE}\left(3\right)$
(position and orientation), and 2) landmark positions $\overline{{\rm p}}=[\overline{{\rm p}}_{1},\overline{{\rm p}}_{2},\ldots,\overline{{\rm p}}_{n}]\in\overline{\mathcal{M}}^{n}$
with the aid of a set of measurements. Fig. \ref{fig:SLAM} presents
a schematic depiction of the SLAM estimation problem. 
\begin{figure}
	\centering{}\includegraphics[scale=0.45]{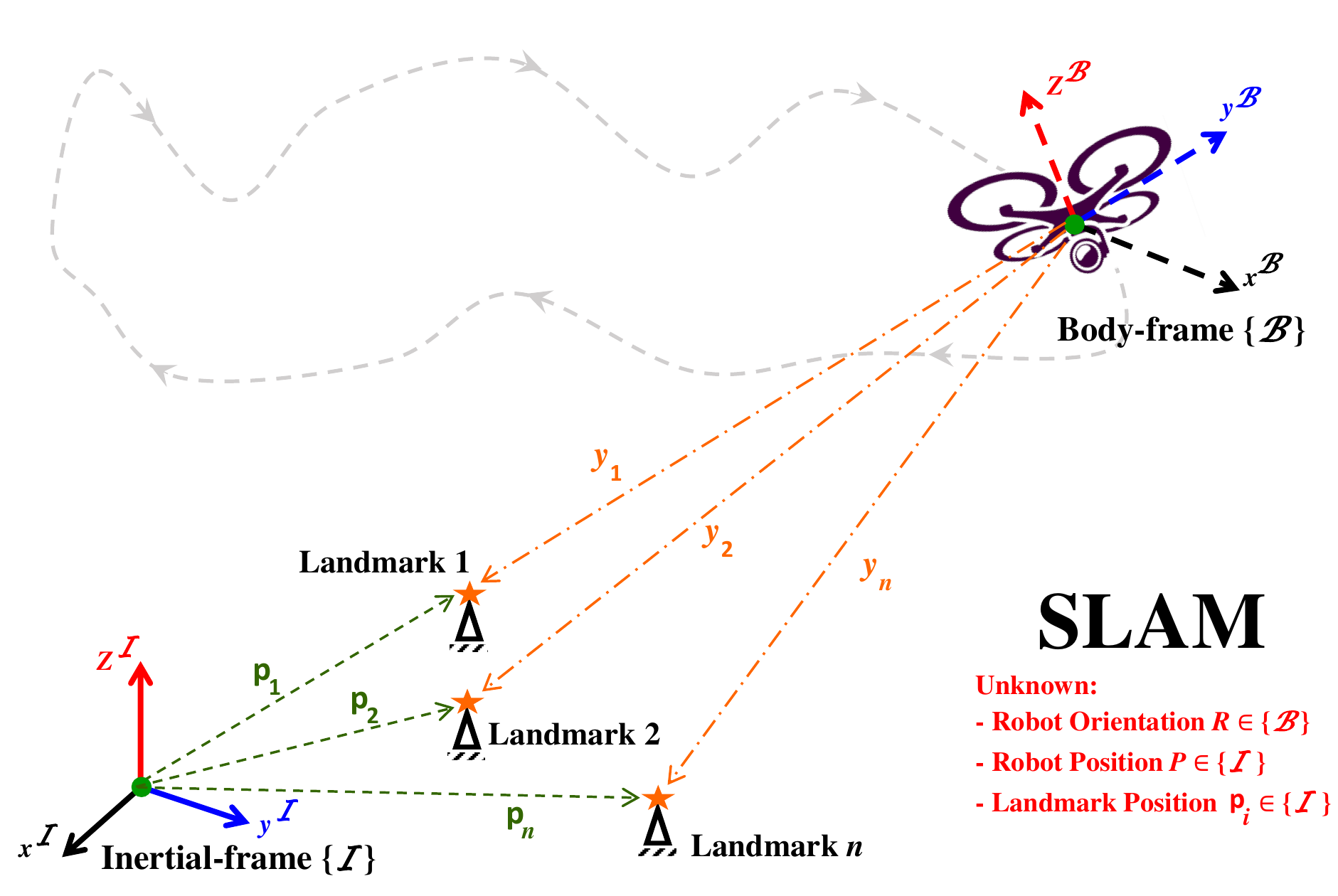}\caption{Graphical illustration of SLAM estimation problem.}
	\label{fig:SLAM}
\end{figure}

From \eqref{eq:T_SLAM}, define $X=(\boldsymbol{T},\overline{{\rm p}})\in\mathbb{SLAM}_{n}\left(3\right)$
as an unknown combination of the true rigid-body's pose and the true
landmark positions. Define $\mathcal{Y}=[\left[U\right]_{\wedge},\overset{\circ}{{\rm v}}]\in\mathfrak{slam}_{n}\left(3\right)$
as the true group velocity where $\overset{\circ}{{\rm v}}=[\overset{\circ}{{\rm v}}_{1},\overset{\circ}{{\rm v}}_{2},\ldots,\overset{\circ}{{\rm v}}_{n}]\in\overset{\circ}{\mathcal{M}}^{n}$
and $U,\overset{\circ}{{\rm v}}\in\left\{ \mathcal{B}\right\} $.
Note that the available measurement $\mathcal{Y}$ is assumed to be
both continuous and bounded. Define $U=\left[\Omega^{\top},V^{\top}\right]^{\top}\in\mathbb{R}^{6}$
as a group velocity vector expressed relative to the body-frame, with
$\Omega\in\mathbb{R}^{3}$ and $V\in\mathbb{R}^{3}$ describing the
true angular and translational velocity, respectively, and ${\rm v}_{i}\in\mathbb{R}^{3}$
describing the $i$th linear velocity of the $i$th landmark with
respect to the body-frame for all $\Omega,V,{\rm v}_{i}\in\left\{ \mathcal{B}\right\} $.
This paper concerns fixed landmarks, signifying that $\dot{{\rm p}}_{i}=\underline{\mathbf{0}}_{3}$,
and therefore, linear velocity of a landmark expressed in the body-frame
as ${\rm v}_{i}=\underline{\mathbf{0}}_{3}$ $\forall i=1,2,\ldots,n$.
The measurements of angular and translational velocities $U_{m}=\left[\Omega_{m}^{\top},V_{m}^{\top}\right]^{\top}\in\mathbb{R}^{6}$
are defined by
\begin{equation}
\begin{cases}
\Omega_{m} & =\Omega+b_{\Omega}\in\mathbb{R}^{3}\\
V_{m} & =V+b_{V}\in\mathbb{R}^{3}
\end{cases}\label{eq:SLAM_TVelcoity}
\end{equation}
with $b_{\Omega}$ and $b_{V}$ being unknown constant bias attached
to the measurements of $\Omega$ and $V$, respectively. Let us define
$b_{U}=\left[b_{\Omega}^{\top},b_{V}^{\top}\right]^{\top}\in\mathbb{R}^{6}$
for all $U_{m},b_{U}\in\left\{ \mathcal{B}\right\} $. The body-frame
vectors relevant to the orientation determination $\overset{\circ}{a}_{j}=[a_{j}^{\top},0]^{\top}$
are described by \cite{hashim2018SO3Stochastic,hashim2019SO3Wiley}
\begin{equation}
\overset{\circ}{a}_{j}=\boldsymbol{T}^{-1}\overset{\circ}{r}_{j}+\overset{\circ}{b}_{j}^{a}+\overset{\circ}{n}_{j}^{a}\in\overset{\circ}{\mathcal{M}},\hspace{1em}j=1,2,\ldots,n_{R}\label{eq:SLAM_Vect_R}
\end{equation}
where $\overset{\circ}{r}_{j}=[r_{j}^{\top},0]^{\top}\in\left\{ \mathcal{I}\right\} $,
$\overset{\circ}{b}_{j}^{a}=[b_{j}^{a^{\top}},0]^{\top}\in\left\{ \mathcal{B}\right\} $,
and $\overset{\circ}{n}_{j}^{a}=[n_{j}^{a^{\top}},0]^{\top}\in\left\{ \mathcal{B}\right\} $
describe the $j$th known inertial-frame observation, unknown constant
bias, and unknown random noise, respectively. Also, $\boldsymbol{T}^{-1}=\mathcal{Z}(R^{\top},-R^{\top}P)$
and let $b_{j}^{a}=n_{j}^{a}=0$. Note that \eqref{eq:SLAM_Vect_R}
exemplifies measurements obtained from a low cost IMU. Both $r_{j}$
and $a_{j}$ in \eqref{eq:SLAM_Vect_R} can be normalized as shown
in \eqref{eq:SLAM_Vector_norm} and employed to establish the orientation
\begin{equation}
\upsilon_{j}^{r}=\frac{r_{j}}{\left\Vert r_{j}\right\Vert },\hspace{1em}\upsilon_{j}^{a}=\frac{a_{j}}{\left\Vert a_{j}\right\Vert }\label{eq:SLAM_Vector_norm}
\end{equation}

\begin{rem}
	\label{rem:R_Marix}\cite{hashim2018SO3Stochastic,hashim2019SO3Wiley}
	The rigid-body's orientation can be established given the availability
	of three non-collinear vectors for both $\upsilon_{j}^{r}$ and $\upsilon_{j}^{a}$.
	If $n_{R}=2$, the third vector can be defined through $\upsilon_{3}^{a}=\upsilon_{1}^{a}\times\upsilon_{2}^{a}$
	and $\upsilon_{3}^{r}=\upsilon_{1}^{r}\times\upsilon_{2}^{r}$, respectively.
\end{rem}
Given an environment with $n$ landmarks, the $i$th landmark measurement
in the body-frame $\overline{y}_{i}=[y_{i}^{\top},1]^{\top}$can be
described by
\begin{equation}
\overline{y}_{i}=\boldsymbol{T}^{-1}\overline{{\rm p}}_{i}+\overset{\circ}{b}_{i}^{y}+\overset{\circ}{n}_{i}^{y}\in\overline{\mathcal{M}},\hspace{1em}\forall i=1,2,\ldots,n\label{eq:SLAM_Vec_Landmark}
\end{equation}
with $\overset{\circ}{b}_{i}^{y}\in\overset{\circ}{\mathcal{M}}$
and $\overset{\circ}{n}_{i}^{y}\in\overset{\circ}{\mathcal{M}}$ describing
unknown constant bias and random noise, respectively.

\begin{assum}\label{Assumption:Feature}\cite{hashim2019SE3Det,hashim2020SE3Stochastic}
	It is assumed that the total number of landmarks available for measurement
	is three or more.\end{assum}

\subsection{True SLAM Kinematics and Error Criteria}

The true SLAM dynamics in \eqref{eq:SLAM_True_dot} are $X=(\boldsymbol{T},\overline{{\rm p}})\in\mathbb{SLAM}_{n}\left(3\right)$,
and the tangent space of $X$ is $\mathcal{Y}=(\left[U\right]_{\wedge},\overset{\circ}{{\rm v}})\in\mathfrak{slam}_{n}\left(3\right)$
defined by
\begin{equation}
\begin{cases}
\dot{\boldsymbol{T}} & =\boldsymbol{T}\left[U\right]_{\wedge}\\
\dot{{\rm p}}_{i} & =R{\rm v}_{i},\hspace{1em}\forall i=1,2,\ldots,n
\end{cases}\label{eq:SLAM_True_dot}
\end{equation}
which can be represented in a simplified form as 
\[
\begin{cases}
\dot{R} & =R\left[\Omega\right]_{\times}\\
\dot{P} & =RV\\
\dot{{\rm p}}_{i} & =R{\rm v}_{i},\hspace{1em}\forall i=1,2,\ldots,n
\end{cases}
\]
where $U=\left[\Omega^{\top},V^{\top}\right]^{\top}$ describes the
true group velocity vector and ${\rm v}_{i}\in\mathbb{R}^{3}$ describes
the $i$th linear velocity of ${\rm p}_{i}$ with respect to the body-frame.
Define the estimate of pose as
\[
\hat{\boldsymbol{T}}=\mathcal{Z}(\hat{R},\hat{P})=\left[\begin{array}{cc}
\hat{R} & \hat{P}\\
\underline{\mathbf{0}}_{3}^{\top} & 1
\end{array}\right]\in\mathbb{SE}\left(3\right)
\]
where $\hat{R}$ represents the estimate of true orientation and $\hat{P}$
stands for the estimate of the true position. Let $\hat{{\rm p}}_{i}$
be the estimate of the true $i$th landmark ${\rm p}_{i}$. Define
error between $\boldsymbol{T}$ and $\hat{\boldsymbol{T}}$ by
\begin{align}
\tilde{\boldsymbol{T}}=\hat{\boldsymbol{T}}\boldsymbol{T}^{-1} & =\left[\begin{array}{cc}
\hat{R} & \hat{P}\\
\underline{\mathbf{0}}_{3}^{\top} & 1
\end{array}\right]\left[\begin{array}{cc}
R^{\top} & -R^{\top}P\\
\underline{\mathbf{0}}_{3}^{\top} & 1
\end{array}\right]\nonumber \\
& =\left[\begin{array}{cc}
\tilde{R} & \tilde{P}\\
\underline{\mathbf{0}}_{3}^{\top} & 1
\end{array}\right]=\mathcal{Z}(\tilde{R},\tilde{P})\label{eq:SLAM_T_error}
\end{align}
where $\tilde{R}=\hat{R}R^{\top}$ denotes the orientation error and
$\tilde{P}=\hat{P}-\tilde{R}P$ describes the position error. The
ultimate goal of pose estimation is to asymptotically drive $\tilde{\boldsymbol{T}}\rightarrow\mathbf{I}_{4}$
and, consequently, drive $\tilde{R}\rightarrow\mathbf{I}_{3}$ and
$\tilde{P}\rightarrow\underline{\mathbf{0}}_{3}$. As such, let the
error between $\hat{{\rm p}}_{i}$ and ${\rm p}_{i}$ be defined as
\begin{equation}
\overset{\circ}{e}_{i}=\overline{\hat{{\rm p}}}_{i}-\tilde{\boldsymbol{T}}\,\overline{{\rm p}}_{i}\label{eq:SLAM_e}
\end{equation}
such that $\overline{\hat{{\rm p}}}_{i}=\left[\hat{{\rm p}}_{i}^{\top},1\right]^{\top}\in\overline{\mathcal{M}}$
and $\overset{\circ}{e}_{i}=\left[e_{i}^{\top},0\right]^{\top}\in\overset{\circ}{\mathcal{M}}$.
Based on the definition of $\tilde{\boldsymbol{T}}$ in \eqref{eq:SLAM_T_error},
$\overset{\circ}{e}_{i}=\overline{\hat{{\rm p}}}_{i}-\hat{\boldsymbol{T}}\boldsymbol{T}^{-1}\,\overline{{\rm p}}_{i}$
and, in view of \eqref{eq:SLAM_Vec_Landmark}, one has
\begin{equation}
\overset{\circ}{e}_{i}=\overline{\hat{{\rm p}}}_{i}-\hat{\boldsymbol{T}}\,\overline{y}_{i}\label{eq:SLAM_e_Final}
\end{equation}
The error in \eqref{eq:SLAM_Vec_Landmark} is represented with respect
to the estimates and vector measurements. Hence, it can be found that
$\overset{\circ}{e}_{i}=\left[\tilde{{\rm p}}_{i}^{\top}-\tilde{P}^{\top},0\right]^{\top}$where
$\tilde{{\rm p}}_{i}=\hat{{\rm p}}_{i}-\tilde{R}{\rm p}_{i}$ describes
the error of the $i$th landmark estimate and $\tilde{P}=\hat{P}-\tilde{R}P$
as defined in \eqref{eq:SLAM_T_error}.

\subsection{IMU Setup}

This work proposes an observer design directly dependent on a group
of measurements. As a result, it is essential to formulate a group
of variables in terms of vector measurements. According to the vectorial
body-frame measurements in \eqref{eq:SLAM_Vect_R} and the related
normalization in \eqref{eq:SLAM_Vector_norm}, define
\begin{equation}
M=M^{\top}=\sum_{j=1}^{n_{{\rm R}}}s_{j}\upsilon_{j}^{r}\left(\upsilon_{j}^{r}\right)^{\top},\hspace{1em}\forall j=1,2,\ldots n_{{\rm R}}\label{eq:SLAM_M}
\end{equation}
with $s_{j}\geq0$ being a constant gain describes the confidence
level of the $j$th sensor measurement. It is evident that $M$ is
symmetric. According to Remark \ref{rem:R_Marix}, it is assumed that
$n_{{\rm R}}\geq2$ such that no less than two non-collinear body-frame
measurements and inertial-frame observations are available. Consequently,
${\rm rank}(M)=3$. Describe the eigenvalues of $M$ as $\lambda_{1}$,
$\lambda_{2}$, and $\lambda_{3}$. Thus, all of $\lambda_{1}$, $\lambda_{2}$,
and $\lambda_{3}$ are positive. Define a new variable $\breve{\mathbf{M}}={\rm Tr}\left\{ M\right\} \mathbf{I}_{3}-M$
provided that ${\rm rank}(M)=3$. As a result, ${\rm rank}(\breve{\mathbf{M}})=3$
and \cite{bullo2004geometric}:
\begin{enumerate}
	\item $\breve{\mathbf{M}}$ is positive-definite.
	\item The three eigenvalues of $\breve{\mathbf{M}}$ are $\lambda_{2}+\lambda_{3}$,
	$\lambda_{3}+\lambda_{1}$, and $\lambda_{2}+\lambda_{1}$, and $\underline{\lambda}(\breve{\mathbf{M}})>0$. 
\end{enumerate}
where $\underline{\lambda}(\cdot)$ denotes the minimum eigenvalue
of a matrix. In the remainder of the paper it is considered that ${\rm rank}\left(M\right)=3$
and $s_{j}$ is selected such that $\sum_{j=1}^{n_{{\rm R}}}s_{j}=3$
for $j=1,2,\ldots n_{{\rm R}}$. This in turn implies that ${\rm Tr}\left\{ M\right\} =3$.
\begin{lem}
	\label{Lemm:SLAM_Lemma1}Let $\tilde{R}\in\mathbb{SO}\left(3\right)$,
	$M=M^{\top}\in\mathbb{R}^{3\times3}$ as in \eqref{eq:SLAM_M} with
	${\rm Tr}\{M\}=3$ and ${\rm rank}\{M\}=3$. Define $\breve{\mathbf{M}}={\rm Tr}\{M\}\mathbf{I}_{3}-M$
	with $\underline{\lambda}=\underline{\lambda}(\breve{\mathbf{M}})$
	being its minimum singular value. Then, the following holds 
	\begin{align}
	||\tilde{R}M||_{{\rm I}} & \leq\frac{2}{\underline{\lambda}}\frac{||\mathbf{vex}(\boldsymbol{\mathcal{P}}_{a}(\tilde{R}M))||^{2}}{1+{\rm Tr}\{\tilde{R}MM^{-1}\}}\label{eq:SLAM_lemm1_2}
	\end{align}
	\textbf{Proof.} See the Appendix in \cite{hashim2019SO3Wiley}.
\end{lem}
\begin{defn}
	\label{def:Unstable-set}Define a non-attractive and forward invariant
	unstable set $\mathcal{U}_{s}\subseteq\mathbb{SO}\left(3\right)$
	as
	\begin{equation}
	\mathcal{U}_{s}=\left\{ \left.\tilde{R}\left(0\right)\in\mathbb{SO}\left(3\right)\right|{\rm Tr}\{\tilde{R}\left(0\right)\}=-1\right\} \label{eq:SO3_PPF_STCH_SET}
	\end{equation}
	$\tilde{R}\left(0\right)\in\mathcal{U}_{s}$ is possible given one
	of the following three scenarios: $\tilde{R}\left(0\right)={\rm diag}(1,-1,-1)$,
	$\tilde{R}\left(0\right)={\rm diag}(-1,1,-1)$, and $\tilde{R}\left(0\right)={\rm diag}(-1,-1,1)$.
\end{defn}
The expressions in \eqref{eq:SLAM_Vect_R} and \eqref{eq:SLAM_Vector_norm}
entail that the true normalized value of the $j$th body-frame vector
is defined by $\upsilon_{j}^{a}=R^{\top}\upsilon_{j}^{r}$. Hence,
define the estimate of the body-frame vector as
\begin{equation}
\hat{\upsilon}_{j}^{a}=\hat{R}^{\top}\upsilon_{j}^{r},\hspace{1em}\forall j=1,2,\ldots n_{{\rm R}}\label{eq:SLAM_vect_R_estimate}
\end{equation}
Define the error in pose as in \eqref{eq:SLAM_T_error} where $\tilde{R}=\hat{R}R^{\top}$.
By the virtue of the identities in \eqref{eq:SLAM_Identity1} and
\eqref{eq:SLAM_Identity2}, one has
\begin{align*}
\left[\hat{R}\sum_{j=1}^{n_{{\rm R}}}\frac{s_{j}}{2}\hat{\upsilon}_{j}^{a}\times\upsilon_{j}^{a}\right]_{\times} & =\hat{R}\sum_{j=1}^{n_{{\rm R}}}\frac{s_{j}}{2}\left(\upsilon_{j}^{a}\left(\hat{\upsilon}_{j}^{a}\right)^{\top}-\hat{\upsilon}_{j}^{a}\left(\upsilon_{j}^{a}\right)^{\top}\right)\hat{R}^{\top}\\
& =\frac{1}{2}\hat{R}R^{\top}M-\frac{1}{2}MR\hat{R}^{\top}\\
& =\boldsymbol{\mathcal{P}}_{a}(\tilde{R}M)
\end{align*}
If follows that $\mathbf{vex}(\boldsymbol{\mathcal{P}}_{a}(\tilde{R}M))$
with respect to vector measurements is as below
\begin{equation}
\boldsymbol{\Upsilon}(\tilde{R}M)=\mathbf{vex}(\boldsymbol{\mathcal{P}}_{a}(\tilde{R}M))=\hat{R}\sum_{j=1}^{n_{{\rm R}}}(\frac{s_{j}}{2}\hat{\upsilon}_{j}^{a}\times\upsilon_{j}^{a})\label{eq:SLAM_VEX_VM}
\end{equation}
Moreover, $\tilde{R}M$ can be specified in terms of vector measurements
as
\begin{equation}
\tilde{R}M=\hat{R}\sum_{j=1}^{n_{{\rm R}}}\left(s_{j}\upsilon_{j}^{a}\left(\upsilon_{j}^{r}\right)^{\top}\right)\label{eq:SLAM_RM_VM}
\end{equation}
Note that ${\rm Tr}\left\{ M\right\} =3$, and, according to the definition
in \eqref{eq:SLAM_Ecul_Dist}, the normalized Euclidean distance of
$\tilde{R}M$ in terms of vector measurements is
\begin{align}
e_{\tilde{R}}=||\tilde{R}M||_{{\rm I}} & =\frac{1}{4}{\rm Tr}\{(\mathbf{I}_{3}-\tilde{R})M\}\nonumber \\
& =\frac{1}{4}{\rm Tr}\left\{ \mathbf{I}_{3}-\hat{R}\sum_{j=1}^{n_{{\rm R}}}\left(s_{j}\upsilon_{j}^{a}\left(\upsilon_{j}^{r}\right)^{\top}\right)\right\} \nonumber \\
& =\frac{1}{4}\sum_{j=1}^{n_{{\rm R}}}\left(1-s_{j}\left(\hat{\upsilon}_{j}^{a}\right)^{\top}\upsilon_{j}^{a}\right)\label{eq:SLAM_RI_VM}
\end{align}
Also, from \eqref{eq:SLAM_Ecul_Dist}, one has 
\begin{align}
1-||\tilde{R}||_{{\rm I}} & =1-\frac{1}{4}{\rm Tr}\{\mathbf{I}_{3}-\tilde{R}\}=\frac{1}{4}(1+{\rm Tr}\{\tilde{R}\})\label{eq:SLAM_Property}
\end{align}
Accordingly, from \eqref{eq:SLAM_Property}, one obtains
\begin{align}
1-||\tilde{R}||_{{\rm I}} & =\frac{1}{4}(1+{\rm Tr}\{\tilde{R}MM^{-1}\})\label{eq:SLAM_property2}
\end{align}
From \eqref{eq:SLAM_property2} and \eqref{eq:SLAM_RM_VM}, one finds
\begin{align}
& \pi(\tilde{R},M)={\rm Tr}\{\tilde{R}MM^{-1}\}\nonumber \\
& \hspace{0.3em}={\rm Tr}\left\{ \left(\sum_{j=1}^{n_{{\rm R}}}s_{j}\upsilon_{j}^{a}\left(\upsilon_{j}^{r}\right)^{\top}\right)\left(\sum_{j=1}^{n_{{\rm R}}}s_{j}\hat{\upsilon}_{j}^{a}\left(\upsilon_{j}^{r}\right)^{\top}\right)^{-1}\right\} \label{eq:SLAM_Gamma_VM}
\end{align}
With the aim of making the observer design adaptive, consider the
estimate of the unknown bias $b_{U}$ to be $\hat{b}_{U}=\left[\hat{b}_{\Omega}^{\top},\hat{b}_{V}^{\top}\right]^{\top}$.
Define the error between $b_{U}$ and $\hat{b}_{U}$ as
\begin{equation}
\begin{cases}
\tilde{b}_{\Omega} & =b_{\Omega}-\hat{b}_{\Omega}\\
\tilde{b}_{V} & =b_{V}-\hat{b}_{V}
\end{cases}\label{eq:SLAM_b_error}
\end{equation}
where $\tilde{b}_{U}=\left[b_{\Omega}^{\top},\tilde{b}_{V}^{\top}\right]^{\top}\in\mathbb{R}^{6}$.
Prior to proceeding, it is crucial to incorporate the nonlinearity
of the true SLAM problem into the proposed observer design. Thus,
the proposed observer follows $\hat{X}=(\hat{\boldsymbol{T}},\overline{\hat{{\rm p}}})\in\mathbb{SLAM}_{n}\left(3\right)$
and its tangent space is $\hat{\mathcal{Y}}=([\hat{U}]_{\wedge},\overset{\circ}{\hat{{\rm v}}})\in\mathfrak{slam}_{n}\left(3\right)$
where $\hat{\boldsymbol{T}}\in\mathbb{SE}\left(3\right)$, $\overline{\hat{{\rm p}}}=[\overline{\hat{{\rm p}}}_{1},\ldots,\overline{\hat{{\rm p}}}_{n}]\in\overline{\mathcal{M}}^{n}$,
$\hat{U}\in\mathfrak{se}\left(3\right)$, and $\overset{\circ}{\hat{{\rm v}}}=[\overset{\circ}{\hat{{\rm v}}}_{1},\ldots,\overset{\circ}{\hat{{\rm v}}}_{n}]\in\overset{\circ}{\mathcal{M}}^{n}$
are the estimates of pose, landmark positions, group velocity, and
linear landmark velocity, respectively.

\section{Nonlinear Observer Design with Systematic Convergence \label{sec:SLAM_Filter}}

This section presents the nonlinear observer design for SLAM with
guaranteed transient and steady-state performance. The nonlinear observer
design is based on the measurements obtained from an IMU and the landmarks
within the environment. 

\subsection{Systematic Convergence\label{subsec:Systematic_Convergence}}

Recall the error in the normalized Euclidean distance $e_{\tilde{R}}$
defined in \eqref{eq:SLAM_RI_VM} and the landmark estimation $e_{i}=[e_{i,1},e_{i,2},e_{i,3}]^{\top}$
defined in \eqref{eq:SLAM_e_Final}. This subsection aims to reformulate
the SLAM estimation problem such that $e_{\tilde{R}}$ and $e_{i}$
follow the predefined transient and steady-state properties adjusted
by the user. The key objective of the reformulation consists in using
predefined reducing boundaries to fully control the process of error
initiation within a given large set and its continuous decay towards
a given small set. This can be achieved via prescribed performance
functions (PPF) \cite{bechlioulis2008robust} defined as positive,
time-decreasing, and smooth functions $\xi_{i,k}:\mathbb{R}_{+}\to\mathbb{R}_{+}$
and $\xi_{\tilde{R}}:\mathbb{R}_{+}\to\mathbb{R}_{+}$ such that
\begin{equation}
\begin{cases}
\xi_{i,k}\left(t\right) & =(\xi_{i,k}^{0}-\xi_{i,k}^{\infty})\exp(-\ell_{i,k}t)+\xi_{i,k}^{\infty}\\
\xi_{\tilde{R}}\left(t\right) & =(\xi_{\tilde{R}}^{0}-\xi_{\tilde{R}}^{\infty})\exp(-\ell_{\tilde{R}}t)+\xi_{\tilde{R}}^{\infty}
\end{cases}\label{eq:SLAM_Presc}
\end{equation}
where $\xi_{i}=[\xi_{i,1},\xi_{i,2},\xi_{i,3}]^{\top}\in\mathbb{R}^{3}$,
$\xi_{\tilde{R}}\left(0\right)=\xi_{\tilde{R}}^{0}$ and $\xi_{i}\left(0\right)=\xi_{i}^{0}=[\xi_{i,1}^{0},\xi_{i,2}^{0},\xi_{i,3}^{0}]^{\top}\in\mathbb{R}^{3}$
represent the initial value and the upper bound of $\xi_{\tilde{R}}$
and $\xi_{i}$, respectively. $\xi_{\tilde{R}}^{\infty}$ and $\xi_{i}^{\infty}=[\xi_{i,1}^{\infty},\xi_{i,2}^{\infty},\xi_{i,3}^{\infty}]^{\top}\in\mathbb{R}^{3}$
correspond to the upper bound of the small sets, $\ell_{\tilde{R}}$
and $\ell_{i}=[\ell_{i,1},\ell_{i,2},\ell_{i,3}]^{\top}\in\mathbb{R}^{3}$
signify positive constants that control the rate of convergence of
$\xi_{\tilde{R}}\left(t\right)$ and $\xi\left(t\right)$, respectively,
for all $i=1,2,\ldots,n$, and $k=1,2,3$. For simplicity, let the
subscript $\star$ denote the appropriate component $\tilde{R}$ or
$i,k$. This way, $e_{\star}\left(t\right)$ is guaranteed to obey
the predefined dynamically reducing boundaries, if one of the following
conditions is fulfilled:
\begin{align}
-\delta_{\star}\xi_{\star}\left(t\right)<e_{\star}\left(t\right)<\xi_{\star}\left(t\right), & \text{ if }e_{\star}\left(0\right)\geq0\label{eq:SLAM_ePos}\\
-\xi_{\star}\left(t\right)<e_{\star}\left(t\right)<\delta_{\star}\xi_{\star}\left(t\right), & \text{ if }e_{\star}\left(0\right)<0\label{eq:SLAM_eNeg}
\end{align}
with $\delta_{\star}\in\left[0,1\right]$. Define $e_{\star}:=e_{\star}\left(t\right)$
and $\xi_{\star}:=\xi_{\star}\left(t\right)$. The systematic convergence
of $e_{\star}$ from a known large set to a known narrow set defined
in \eqref{eq:SLAM_ePos} and \eqref{eq:SLAM_eNeg} is presented in
Fig. \ref{fig:SO3PPF_2}.

\begin{figure}[h!]
	\centering{}\includegraphics[scale=0.32]{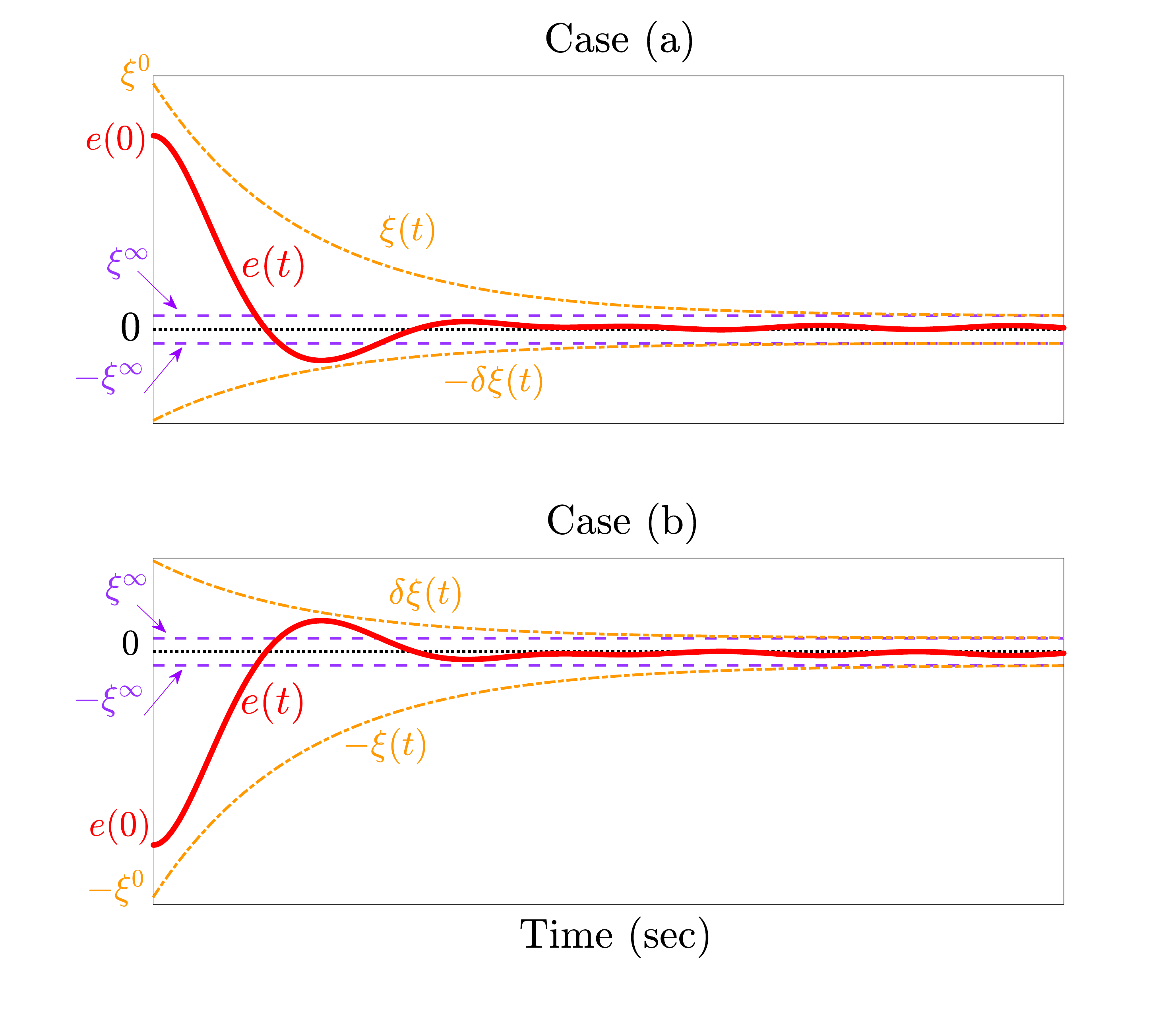} \caption{Systematic convergence of $e_{\star}$ with respect to (a) Equation
		\eqref{eq:SLAM_ePos}; (b) Equation \eqref{eq:SLAM_eNeg}.}
	\label{fig:SO3PPF_2}
\end{figure}

\begin{rem}
	\label{SE3PPF_rem3}\cite{bechlioulis2008robust,hashim2019SO3Wiley,hashim2019SE3Det,wei2018learning,wei2018novel}
	Note that specifying the upper bound as well as the sign of $e_{\star}\left(0\right)$
	is sufficient to compel the error to adhere to the performance constraints
	and obey the predefined dynamically reducing boundaries $\forall t>0$.
	Given that one of the conditions in \eqref{eq:SLAM_ePos} and \eqref{eq:SLAM_eNeg}
	is met, the maximum undershoot/overshoot will be bounded by $\pm\delta\xi_{\star}$
	and the steady-state error will be limited by $\pm\xi_{\star}^{\infty}$as
	detailed in Fig. \ref{fig:SO3PPF_2}. 
\end{rem}
Let the error $e_{\star}$ be given as
\begin{equation}
e_{\star}=\xi_{\star}\mathcal{F}(E_{\star})\label{eq:SLAM_e_Trans}
\end{equation}
where $\xi_{\star}\in\mathbb{R}$ is as expressed in \eqref{eq:SLAM_Presc},
$E_{\star}\in\mathbb{R}$ represents the unconstrained error known
as transformed error, and $\mathcal{F}(E_{\star})$ represents a smooth
function and follows Assumption \ref{Assum:SE3PPF_1}:

\begin{assum}\label{Assum:SE3PPF_1} The smooth function $\mathcal{F}(E_{\star})$
	has the following characteristics \cite{bechlioulis2008robust,hashim2019SO3Wiley,hashim2019SE3Det}: 
	\begin{enumerate}
		\item[\textbf{1.}] $\mathcal{F}(E_{\star})$ strictly increasing. 
		\item[\textbf{2.}] $\mathcal{F}(E_{\star})$ is constrained by\\
		$-\underline{\delta}_{\star}<\mathcal{F}(E_{\star})<\bar{\delta}_{\star},{\rm \text{ if }}e_{\star}\left(0\right)\geq0$\\
		$-\bar{\delta}_{\star}<\mathcal{F}(E_{\star})<\underline{\delta}_{\star},{\rm \text{ if }}e_{\star}\left(0\right)<0$
		\\
		where $\bar{\delta}_{\star}$ and $\underline{\delta}_{\star}$ are
		positive constants with $\underline{\delta}_{\star}\leq\bar{\delta}_{\star}$. 
		\item[\textbf{3.}] 
		\item[] $\left.\begin{array}{c}
		\underset{E_{\star}\rightarrow-\infty}{\lim}\mathcal{F}(E_{\star})=-\underline{\delta}_{\star}\\
		\underset{E_{\star}\rightarrow+\infty}{\lim}\mathcal{F}(E_{\star})=\bar{\delta}_{\star}
		\end{array}\right\} {\rm \text{ if }}e_{\star}\left(0\right)\geq0$\\
		$\left.\begin{array}{c}
		\underset{E_{\star}\rightarrow-\infty}{\lim}\mathcal{F}(E_{\star})=-\bar{\delta}_{\star}\\
		\underset{E_{\star}\rightarrow+\infty}{\lim}\mathcal{F}(E_{\star})=\underline{\delta}_{\star}
		\end{array}\right\} {\rm \text{ if }}e_{\star}\left(0\right)<0$ \\
		such that 
	\end{enumerate}
	\begin{equation}
	\mathcal{F}\left(E_{\star}\right)=\begin{cases}
	\frac{\bar{\delta}_{\star}\exp(E_{\star})-\underline{\delta}_{\star}\exp(-E_{\star})}{\exp(E_{\star})+\exp(-E_{\star})}, & \bar{\delta}_{\star}\geq\underline{\delta}_{\star}\text{ if }e_{\star}\left(0\right)\geq0\\
	\frac{\bar{\delta}_{\star}\exp(E_{\star})-\underline{\delta}_{\star}\exp(-E_{\star})}{\exp(E_{\star})+\exp(-E_{\star})}, & \underline{\delta}_{\star}\geq\bar{\delta}_{\star}\text{ if }e_{\star}\left(0\right)<0
	\end{cases}\label{eq:SLAM_Smooth}
	\end{equation}
	
	$E_{\star}$ can be expressed via inverse transformation of \eqref{eq:SLAM_Smooth}
	\begin{equation}
	E_{\star}(e_{\star},\xi_{\star})=\mathcal{F}^{-1}(e_{\star}/\xi_{\star})\label{eq:SLAM_Trans1}
	\end{equation}
\end{assum}

In view of \eqref{eq:SLAM_Smooth}, the inverse transformation of
$\mathcal{F}\left(E_{\star}\right)$ is given by
\begin{equation}
\begin{aligned}E_{\star}= & \frac{1}{2}\begin{cases}
\text{ln}\frac{\underline{\delta}_{\star}+e_{\star}/\xi_{\star}}{\bar{\delta}_{\star}-e_{\star}/\xi_{\star}}, & \bar{\delta}_{\star}\geq\underline{\delta}_{\star}\text{ if }e_{\star}\left(0\right)\geq0\\
\text{ln}\frac{\underline{\delta}_{\star}+e_{\star}/\xi_{\star}}{\bar{\delta}_{\star}-e_{\star}/\xi_{\star}}, & \underline{\delta}_{\star}\geq\bar{\delta}_{\star}\text{ if }e_{\star}\left(0\right)<0
\end{cases}\end{aligned}
\label{eq:SLAM_trans2}
\end{equation}
\begin{comment}
\label{rem:SO3PPF_1} \cite{bechlioulis2008robust}The transient
and steady-state performance of $e_{\star}$ is constrained by $\xi_{i}$,
and thereby, the guaranteed performance is attained if and only if
the transformed error $E_{\star}$ is guaranteed to be bounded $\forall t\geq0$.
\end{comment}

\begin{prop}
	\label{Prop:SLAM_1}Consider $e_{\tilde{R}}$ in \eqref{eq:SLAM_RI_VM}
	and $e_{i,k}$ in \eqref{eq:SLAM_e_Final}. Recall \eqref{eq:SLAM_e_Trans},
	\eqref{eq:SLAM_Smooth}, and \eqref{eq:SLAM_Trans1}, and let the
	transformed error of $E_{\tilde{R}}$ and $E_{i,k}$ be described
	as in \eqref{eq:SLAM_trans2} given that $\underline{\delta}_{\tilde{R}}=\bar{\delta}_{\tilde{R}}$
	and $\underline{\delta}_{i,k}=\bar{\delta}_{i,k}$. Then, the following
	holds: 
\end{prop}
\begin{enumerate}
	\item[(i)] The only possible representation of $E_{\tilde{R}}$ is
	\begin{equation}
	E_{\tilde{R}}=\frac{1}{2}\text{ln}\frac{\bar{\delta}_{\tilde{R}}+e_{\tilde{R}}/\xi_{\tilde{R}}}{\underline{\delta}_{\tilde{R}}-e_{\tilde{R}}/\xi_{\tilde{R}}}\label{eq:SLAM_trans3}
	\end{equation}
	\item[(ii)] The transformed error $E_{\tilde{R}}>0\forall e_{\tilde{R}}\neq0$
	and $E_{i,k}\neq0\forall e_{i,k}\neq0$.
	\item[(iii)] $E_{\star}=0$ only at $e_{\star}=0$.
\end{enumerate}
\begin{proof}Due to the fact that $e_{\tilde{R}}=||\tilde{R}M||_{{\rm I}}=\frac{1}{4}{\rm Tr}\{(\mathbf{I}_{3}-\tilde{R})M\}$
	and $M$ is positive definite, $e_{\tilde{R}}>0\forall\tilde{R}\ensuremath{\neq}\mathbf{I}_{3}$
	and $e_{\tilde{R}}=0$ only at $\tilde{R}=\mathbf{I}_{3}$. As such,
	the upper portion of \eqref{eq:SLAM_trans2} holds $\forall t\geq0$
	which proves (i). Since Proposition \ref{Prop:SLAM_1} states that
	$\underline{\delta}_{\tilde{R}}=\bar{\delta}_{\tilde{R}}$ and $\underline{\delta}_{i,k}=\bar{\delta}_{i,k}$
	with the constraint of $e_{\tilde{R}}\leq\xi_{\tilde{R}}$ and $e_{i,k}\leq\xi_{i,k}$,
	the expression in \eqref{eq:SLAM_trans3} is $(\bar{\delta}_{\tilde{R}}+e_{\tilde{R}}/\xi_{\tilde{R}})/(\underline{\delta}_{\tilde{R}}-e_{\tilde{R}}/\xi_{\tilde{R}})\geq1\forall e_{\tilde{R}}\neq0$
	and $(\bar{\delta}_{i,k}+e_{i,k}/\xi_{i,k})/(\underline{\delta}_{i,k}-e_{i,k}/\xi_{i,k})\neq1\forall e_{i,k}\neq0$.
	Moreover, it logically follows that at $e_{\star}=0$, $E_{\star}=\frac{1}{2}\text{ln}(\bar{\delta}_{\star}/\underline{\delta}_{\star})=0$.
	This proves (ii) and (iii).\end{proof} Consider
\begin{align*}
\Lambda_{\star} & =\frac{1}{2\xi_{\star}}\frac{\partial\mathcal{F}^{-1}(e_{\star}/\xi_{\star})}{\partial(e_{\star}/\xi_{\star})}\\
& =\frac{1}{2\xi_{\star}}(\frac{1}{\underline{\delta}_{\star}+e_{\star}/\xi_{\star}}+\frac{1}{\bar{\delta}_{\star}-e_{\star}/\xi_{\star}})
\end{align*}
Define the following variables
\begin{equation}
\begin{cases}
\mu_{\tilde{R}} & =\dot{\xi}_{\tilde{R}}/\xi_{\tilde{R}}\\
\mu_{i} & ={\rm diag}\left(\dot{\xi}_{i,1}/\xi_{i,1},\dot{\xi}_{i,2}/\xi_{i,2},\dot{\xi}_{i,3}/\xi_{i,3}\right)\\
\Lambda_{i} & ={\rm diag}(\Lambda_{i,1},\Lambda_{i,2},\Lambda_{i,3})
\end{cases}\label{eq:SLAM_Aux1}
\end{equation}
for $i=1,2,\ldots,n$. Thus, one finds that the transformed error
dynamics of $E_{\tilde{R}}$ and $E_{i}=\left[E_{i,1},E_{i,2},E_{i,3}\right]^{\top}$
are as follows:
\begin{equation}
\begin{cases}
\dot{E}_{\tilde{R}} & =\Lambda_{\tilde{R}}(\dot{e}_{\tilde{R}}-\mu_{\tilde{R}}e_{\tilde{R}})\\
\dot{E}_{i} & =\Lambda_{i}(\dot{e}_{i}-\mu_{i}e_{i})
\end{cases}\label{eq:SLAM_Trans_dot}
\end{equation}

\subsection{Nonlinear Observer Design\label{subsec:Det_without_IMU}}

Consider the following nonlinear observer
\begin{equation}
\begin{cases}
\dot{\hat{\boldsymbol{T}}} & =\hat{\boldsymbol{T}}\left[U_{m}-\hat{b}_{U}-W_{U}\right]_{\wedge}\\
\dot{{\rm \hat{p}}}_{i} & =-k_{1}\Lambda_{i}E_{i}+\hat{R}\left[y_{i}\right]_{\times}W_{\Omega},\hspace{1em}i=1,2,\ldots,n\\
\dot{\hat{b}}_{\Omega} & =\frac{\Lambda_{\tilde{R}}}{2}\Gamma_{1}\hat{R}^{\top}\boldsymbol{\Upsilon}(\tilde{R}M)-\sum_{i=1}^{n}\frac{\Gamma_{2}}{\alpha_{i}}\left[y_{i}\right]_{\times}\hat{R}^{\top}\Lambda_{i}E_{i}\\
\dot{\hat{b}}_{V} & =-\sum_{i=1}^{n}\frac{\Gamma_{2}}{\alpha_{i}}\hat{R}^{\top}\Lambda_{i}E_{i}\\
\tau_{R} & =\underline{\lambda}(\breve{\mathbf{M}})\times\left(1+\pi(\tilde{R},M)\right)\\
W_{\Omega} & =\frac{k_{w}\Lambda_{\tilde{R}}-4\mu_{\tilde{R}}}{\tau_{R}}\hat{R}^{\top}\boldsymbol{\Upsilon}(\tilde{R}M)\\
W_{V} & =-\sum_{i=1}^{n}\frac{1}{\alpha_{i}}k_{2}\hat{R}^{\top}\Lambda_{i}E_{i}
\end{cases}\label{eq:SLAM_Filter}
\end{equation}
where $W_{U}=\left[W_{\Omega}^{\top},W_{V}^{\top}\right]^{\top}\in\mathbb{R}^{6}$
is a correction factor and $\hat{b}_{U}$ is an estimate of $b_{U}$.
$k_{1}$, $k_{2}$, $k_{w}$, $\Gamma=\left[\begin{array}{cc}
\Gamma_{1} & 0_{3\times3}\\
0_{3\times3} & \Gamma_{2}
\end{array}\right]$, and $\alpha_{i}$ denote positive constants. Note that $M$, $\pi(\tilde{R},M)$,
$\boldsymbol{\Upsilon}(\tilde{R}M)$, and $e_{i}$ are defined with
respect to the available measurements in \eqref{eq:SLAM_M}, \eqref{eq:SLAM_Gamma_VM},
\eqref{eq:SLAM_VEX_VM}, and \eqref{eq:SLAM_e_Final}, respectively,
for all $i=1,2,\ldots,n$.
\begin{thm}
	\label{thm:PPF}Consider the SLAM kinematics $\dot{X}=(\dot{\boldsymbol{T}},\dot{\overline{{\rm p}}})$
	in \eqref{eq:SLAM_True_dot} in combination with the measurements
	extracted from landmarks ($\overline{y}_{i}=\boldsymbol{T}^{-1}\overline{{\rm p}}_{i}$),
	inertial measurement units $\upsilon_{j}^{a}=R^{\top}\upsilon_{j}^{r}$,
	and velocity measurements ($U_{m}=U+b_{U}$) for all $i=1,2,\ldots,n$
	and $j=1,2,\ldots n_{{\rm R}}$. Let Assumption \ref{Assumption:Feature}
	hold and the observer design be as in \eqref{eq:SLAM_Filter}, supplied
	with the measurements $U_{m}$, $\upsilon_{j}^{a}$ and $\overline{y}_{i}$.
	Let the design parameters $\xi_{\star}^{0}\geq e_{\star}\left(0\right)$,
	$\underline{\delta}_{\star}=\bar{\delta}_{\star}$, $k_{1}$, $k_{2}$,
	$k_{w}$, $\Gamma$, and $\alpha_{i}$ be selected as positive constants.
	Define the following set:
	\begin{align}
	\mathcal{S}= & \left\{ \left(E_{\tilde{R}},E_{1},E_{2},\ldots,E_{n}\right)\in\mathbb{R}\times\mathbb{R}^{3}\times\mathbb{R}^{3}\times\cdots\times\mathbb{R}^{3}\right|\nonumber \\
	& \hspace{5em}\left.E_{\tilde{R}}=0,E_{i}=\underline{\mathbf{0}}_{3}\forall i=1,2,\ldots n\right\} \label{eq:SLAM_Set2}
	\end{align}
	with $E_{\tilde{R}}\left(0\right)\ensuremath{\in\mathcal{L}_{\infty}}$,
	$E_{i}\left(0\right)\ensuremath{\in\mathcal{L}_{\infty}}$ and $\tilde{R}\left(0\right)\notin\mathcal{U}_{s}$.
	Then, 1) all signals in the closed loop are bounded, 2) the error
	$(E_{\tilde{R}},E_{1},E_{2},\ldots,E_{n})$ asymptotically approaches
	$\mathcal{S}$, 3) the error $(e_{\tilde{R}},e_{1},e_{2},\ldots,e_{n})$
	asymptotically approaches the origin, 4) $\tilde{R}$ asymptotically
	approaches the attractive equilibrium point $\mathbf{I}_{3}$, and
	5) $\tilde{b}_{U}$ asymptotically approaches the origin.
\end{thm}
\begin{proof}In order to prove Theorem \ref{thm:PPF}, it is necessary
	to derive transformed error dynamics $\dot{E}_{\star}=\dot{E}_{\star}(e_{\star},\dot{e}_{\star})$
	in \eqref{eq:SLAM_Trans_dot}. It becomes evident that $\dot{E}_{\star}$
	is a function of the pose and the landmark error as well as the error
	dynamics. As such, the first step is to find the pose and landmark
	error dynamics. Recall the pose error $\tilde{\boldsymbol{T}}$ defined
	in \eqref{eq:SLAM_T_error}. Thus, the pose error dynamics are
	\begin{align}
	\dot{\tilde{\boldsymbol{T}}} & =\dot{\hat{\boldsymbol{T}}}\boldsymbol{T}^{-1}+\hat{\boldsymbol{T}}\dot{\boldsymbol{T}}^{-1}\nonumber \\
	& =\hat{\boldsymbol{T}}\left[U+\tilde{b}_{U}-W_{U}\right]_{\wedge}\boldsymbol{T}^{-1}-\hat{\boldsymbol{T}}\left[U\right]_{\wedge}\boldsymbol{T}^{-1}\nonumber \\
	& =\hat{\boldsymbol{T}}\left[\tilde{b}_{U}-W_{U}\right]_{\wedge}\hat{\boldsymbol{T}}^{-1}\tilde{\boldsymbol{T}}\label{eq:SLAM_T_error_dot}
	\end{align}
	where $\boldsymbol{\dot{T}}^{-1}=-\boldsymbol{T}^{-1}\boldsymbol{\dot{T}}\boldsymbol{T}^{-1}$.
	Thereby, the error dynamics of $\overset{\circ}{e}_{i}$ in \eqref{eq:SLAM_e}
	become
	\begin{align}
	\overset{\circ}{\dot{e}}_{i} & =\overset{\circ}{\dot{\hat{{\rm p}}}}_{i}-\dot{\tilde{\boldsymbol{T}}}\,\overline{{\rm p}}_{i}-\tilde{\boldsymbol{T}}\,\dot{\overline{{\rm p}}}_{i}\nonumber \\
	& =\overset{\circ}{\dot{\hat{{\rm p}}}}_{i}-\hat{\boldsymbol{T}}\left[\tilde{b}_{U}-W_{U}\right]_{\wedge}\hat{\boldsymbol{T}}^{-1}\tilde{\boldsymbol{T}}\,\overline{{\rm p}}_{i}\label{eq:SLAM_e_dot}
	\end{align}
	The expression $\hat{\boldsymbol{T}}\left[\tilde{b}_{U}\right]_{\wedge}\hat{\boldsymbol{T}}^{-1}$
	in \eqref{eq:SLAM_T_error_dot} can be reformulated as
	\begin{align}
	\hat{\boldsymbol{T}}\left[\tilde{b}_{U}\right]_{\wedge}\hat{\boldsymbol{T}}^{-1} & =\left[\begin{array}{cc}
	\hat{R}[\tilde{b}_{\Omega}]_{\times}\hat{R}^{\top} & \hat{R}\tilde{b}_{V}-\hat{R}[\tilde{b}_{\Omega}]_{\times}\hat{R}^{\top}\hat{P}\\
	\underline{\mathbf{0}}_{3}^{\top} & 0
	\end{array}\right]\nonumber \\
	& =\left[\begin{array}{c}
	\hat{R}\tilde{b}_{\Omega}\\
	\hat{R}\tilde{b}_{V}+[\hat{P}]_{\times}\hat{R}\tilde{b}_{\Omega}
	\end{array}\right]_{\wedge}\in\mathfrak{se}\left(3\right)\label{eq:SLAM_Adj_Property4}
	\end{align}
	with $[\hat{R}\tilde{b}_{\Omega}]_{\times}=\hat{R}[\tilde{b}_{\Omega}]_{\times}\hat{R}^{\top}$
	as defined in \eqref{eq:SLAM_Identity1}, which shows that the expression
	in \eqref{eq:SLAM_Adj_Property4} is equivalent to
	\begin{equation}
	\hat{\boldsymbol{T}}\left[\tilde{b}_{U}\right]_{\wedge}\hat{\boldsymbol{T}}^{-1}=\left[\left[\begin{array}{cc}
	\hat{R} & 0_{3\times3}\\
	\left[\hat{P}\right]_{\times}\hat{R} & \hat{R}
	\end{array}\right]\tilde{b}_{U}\right]_{\wedge}\label{eq:SLAM_Adj_Property5}
	\end{equation}
	Hence, in view of \eqref{eq:SLAM_Adj_Property5} and \eqref{eq:SLAM_e_dot},
	one obtains{\small{}
		\begin{align}
		\hat{\boldsymbol{T}}\left[\tilde{b}_{U}\right]_{\wedge}\hat{\boldsymbol{T}}^{-1}\tilde{\boldsymbol{T}}\,\overline{{\rm p}}_{i} & =\left[\left[\begin{array}{cc}
		\hat{R} & 0_{3\times3}\\
		\left[\hat{P}\right]_{\times}\hat{R} & \hat{R}
		\end{array}\right]\tilde{b}_{U}\right]_{\wedge}\left[\begin{array}{c}
		\hat{R}y_{i}+\hat{P}\\
		1
		\end{array}\right]\nonumber \\
		& =\left[\begin{array}{c}
		-\left[\hat{R}y_{i}\right]_{\times}\hat{R}\tilde{b}_{\Omega}+\hat{R}\tilde{b}_{V}\\
		0
		\end{array}\right]\nonumber \\
		& =\left[\begin{array}{cc}
		-\hat{R}\left[y_{i}\right]_{\times} & \hat{R}\\
		\underline{\mathbf{0}}_{3}^{\top} & \underline{\mathbf{0}}_{3}^{\top}
		\end{array}\right]\tilde{b}_{U}\label{eq:SLAM_Adj_Property6}
		\end{align}
	}From \eqref{eq:SLAM_e_dot} and \eqref{eq:SLAM_Adj_Property6}, the
	error dynamics in \eqref{eq:SLAM_e_dot} become
	\begin{align}
	\overset{\circ}{\dot{e}}_{i} & =\overset{\circ}{\dot{\hat{{\rm p}}}}_{i}-\left[\begin{array}{cc}
	-\hat{R}\left[y_{i}\right]_{\times} & \hat{R}\\
	\underline{\mathbf{0}}_{3}^{\top} & \underline{\mathbf{0}}_{3}^{\top}
	\end{array}\right]\left(\tilde{b}_{U}-W_{U}\right)\label{eq:SLAM_e_dot1}
	\end{align}
	Note that the last row of \eqref{eq:SLAM_e_dot1} is comprised of
	zeros which means that
	\begin{align}
	\dot{e}_{i} & =\dot{\hat{{\rm p}}}_{i}-[\begin{array}{cc}
	-\hat{R}\left[y_{i}\right]_{\times} & \hat{R}\end{array}](\tilde{b}_{U}-W_{U})\label{eq:SLAM_e_dot_Final}
	\end{align}
	From \eqref{eq:SLAM_T_error_dot}, the attitude error dynamics are
	\begin{align}
	\dot{\tilde{R}} & =\dot{\hat{R}}R^{\top}+\hat{R}\dot{R}^{\top}=\hat{R}\left[\tilde{b}_{\Omega}-W_{\Omega}\right]_{\times}R^{\top}\nonumber \\
	& =\left[\hat{R}(\tilde{b}_{\Omega}-W_{\Omega})\right]_{\times}\tilde{R}\label{eq:SLAM_Attit_Error_dot-1}
	\end{align}
	where the identity in \eqref{eq:SLAM_Identity1} was utilized. In
	view of \eqref{eq:SLAM_Ecul_Dist}, $e_{\tilde{R}}=\frac{1}{4}{\rm Tr}\left\{ (\mathbf{I}_{3}-\tilde{R})M\right\} $.
	As such, with the aid of \eqref{eq:SLAM_Identity6} one has
	\begin{align}
	\dot{e}_{\tilde{R}}= & -\frac{1}{4}{\rm Tr}\left\{ \left[\hat{R}(\tilde{b}_{\Omega}-W_{\Omega})\right]_{\times}\tilde{R}M\right\} \nonumber \\
	= & -\frac{1}{4}{\rm Tr}\left\{ \tilde{R}M\boldsymbol{\mathcal{P}}_{a}\left(\left[\hat{R}(\tilde{b}_{\Omega}-W_{\Omega})\right]_{\times}\right)\right\} \nonumber \\
	= & \frac{1}{2}\mathbf{vex}\left(\boldsymbol{\mathcal{P}}_{a}(\tilde{R}M)\right)^{\top}\hat{R}(\tilde{b}_{\Omega}-W_{\Omega})\label{eq:SLAM_RI_VM_dot_Final}
	\end{align}
	To this end, from \eqref{eq:SLAM_e_dot_Final} and \eqref{eq:SLAM_RI_VM_dot_Final},
	the transformed error dynamics can be represented as
	\begin{equation}
	\begin{cases}
	\dot{E}_{\tilde{R}} & =\Lambda_{\tilde{R}}\left(\frac{1}{2}\boldsymbol{\Upsilon}(\tilde{R}M)^{\top}\hat{R}(\tilde{b}_{\Omega}-W_{\Omega})-\mu_{\tilde{R}}e_{\tilde{R}}\right)\\
	\dot{E}_{i} & =\Lambda_{i}\left(\dot{\hat{{\rm p}}}_{i}-[\begin{array}{cc}
	-\hat{R}\left[y_{i}\right]_{\times} & \hat{R}\end{array}](\tilde{b}_{U}-W_{U})-\mu_{i}e_{i}\right)
	\end{cases}\label{eq:SLAM_E_dot_Final}
	\end{equation}
	Note that $\mu_{\tilde{R}}$ and $\mu_{i}$ are vanishing components,
	and therefore, $\mu_{\tilde{R}},\mu_{i}\rightarrow0$ as $t\rightarrow\infty$.
	Define the following candidate Lyapunov function $V=V(E_{\tilde{R}},E_{1},E_{2},\ldots,E_{n},\tilde{b}_{U})$
	\begin{equation}
	V=E_{\tilde{R}}+\sum_{i=1}^{n}\frac{1}{4\alpha_{i}}\left\Vert E_{i}\right\Vert ^{2}+\frac{1}{2}\tilde{b}_{U}^{\top}\Gamma^{-1}\tilde{b}_{U}\label{eq:SLAM_Lyap1}
	\end{equation}
	From \eqref{eq:SLAM_E_dot_Final} and \eqref{eq:SLAM_Lyap1}, the
	time derivative of $V$ is
	\begin{align}
	& \dot{V}=\dot{E}_{\tilde{R}}+\sum_{i=1}^{n}\frac{1}{\alpha_{i}}E_{i}^{\top}\dot{E}_{i}-\tilde{b}_{U}^{\top}\Gamma^{-1}\dot{\hat{b}}_{U}\nonumber \\
	& =\sum_{i=1}^{n}\frac{1}{\alpha_{i}}\left[\begin{array}{c}
	\Lambda_{\tilde{R}}\boldsymbol{\Upsilon}(\tilde{R}M)\\
	\Lambda_{i}E_{i}
	\end{array}\right]^{\top}\left[\begin{array}{cc}
	\frac{\alpha_{i}}{2}\hat{R} & 0_{3\times3}\\
	\hat{R}[y_{i}]_{\times} & -\hat{R}
	\end{array}\right](\tilde{b}_{U}-W_{U})\nonumber \\
	& \hspace{0.5em}+\sum_{i=1}^{n}\frac{1}{\alpha_{i}}E_{i}^{\top}\Lambda_{i}\left(\dot{\hat{{\rm p}}}_{i}-\mu_{i}e_{i}\right)-\Lambda_{\tilde{R}}\mu_{\tilde{R}}e_{\tilde{R}}-\tilde{b}_{U}^{\top}\Gamma^{-1}\dot{\hat{b}}_{U}\label{eq:SLAM_Lyap2_dot2}
	\end{align}
	By \eqref{eq:SLAM_e_Trans} and \eqref{eq:SLAM_trans2} $|e_{i,k}|\leq\mu_{i,k}\bar{\delta}_{i,k}\xi_{i,k}|E_{i,k}|$,
	and moreover, $\mu_{i,k}$ is a vanishing component. Replacing $W_{U}$,
	$\dot{\hat{b}}_{U}$, and $\dot{\hat{{\rm p}}}_{i}$ with their definition
	in \eqref{eq:SLAM_Filter}, one has
	\begin{align}
	\dot{V}= & -\sum_{i=1}^{n}\frac{k_{1}}{\alpha_{i}}\left\Vert \Lambda_{i}E_{i}\right\Vert ^{2}-\frac{k_{w}\Lambda_{\tilde{R}}^{2}}{2\tau_{R}}\left\Vert \boldsymbol{\Upsilon}(\tilde{R}M)\right\Vert ^{2}\nonumber \\
	& -k_{2}\left\Vert \sum_{i=1}^{n}\frac{1}{\alpha_{i}}\Lambda_{i}E_{i}\right\Vert ^{2}\label{eq:SLAM_Lyap2_dot4}
	\end{align}
	In view of the result of Lemma \ref{Lemm:SLAM_Lemma1}, $2||\boldsymbol{\Upsilon}(\tilde{R}M)||^{2}/\tau_{R}\geq e_{\tilde{R}}$.
	Also, from \eqref{eq:SLAM_Smooth}, $e_{\tilde{R}}=\xi_{\tilde{R}}\frac{\bar{\delta}_{\tilde{R}}\exp(E_{\tilde{R}})-\underline{\delta}_{\tilde{R}}\exp(-E_{\tilde{R}})}{\exp(E_{\tilde{R}})+\exp(-E_{\tilde{R}})}$.
	Hence, the expression in \eqref{eq:SLAM_Lyap2_dot4} becomes
	\begin{align}
	\dot{V}\leq & -\frac{k_{w}}{8}\Lambda_{\tilde{R}}^{2}\xi_{\tilde{R}}\frac{\bar{\delta}_{\tilde{R}}\exp(E_{\tilde{R}})-\underline{\delta}_{\tilde{R}}\exp(-E_{\tilde{R}})}{\exp(E_{\tilde{R}})+\exp(-E_{\tilde{R}})}\nonumber \\
	& -\sum_{i=1}^{n}\frac{k_{1}}{\alpha_{i}}\left\Vert \Lambda_{i}E_{i}\right\Vert ^{2}-k_{2}\left\Vert \sum_{i=1}^{n}\frac{1}{\alpha_{i}}\Lambda_{i}E_{i}\right\Vert ^{2}\label{eq:SLAM_Lyap2_final}
	\end{align}
	From \eqref{eq:SLAM_Lyap2_final}, given that $E_{\tilde{R}}\left(0\right)\ensuremath{\in\mathcal{L}_{\infty}}$,
	$E_{i}\left(0\right)\ensuremath{\in\mathcal{L}_{\infty}}$, and $\tilde{R}\left(0\right)\notin\mathcal{U}_{s}$,
	$\dot{V}<0$ for all $E_{i}\neq\underline{\mathbf{0}}_{3}$ or $E_{\tilde{R}}\neq0$,
	and $\dot{V}=0$ only at $E_{i}=\underline{\mathbf{0}}_{3}$ and $E_{\tilde{R}}=0$
	for all $i=1,2,\ldots,n$. Accordingly, the inequality in \eqref{eq:SLAM_Lyap2_final}
	ensures that $E_{i}$ and $E_{\tilde{R}}$ asymptotically converge
	to the set $\mathcal{S}$ in \eqref{eq:SLAM_Set2}. Also, $\dot{V}$
	is negative, continuous, and $\dot{V}\rightarrow0$ such that $V\in\mathcal{L}_{\infty}$
	and a finite $\lim_{t\rightarrow\infty}V$ exists. According to (iii)
	in Proposition \ref{Prop:SLAM_1}, $E_{i}\neq\underline{\mathbf{0}}_{3}$
	for $e_{i}\neq0$, $E_{\tilde{R}}\neq0$ for $e_{\tilde{R}}\neq0$,
	$E_{i}=\underline{\mathbf{0}}_{3}$ only at $e_{i}=0$, while $E_{\tilde{R}}=0$
	only at $e_{\tilde{R}}=0$, respectively. Thus, $e_{i}\rightarrow0$
	and $e_{\tilde{R}}=||\tilde{R}M||_{{\rm I}}\rightarrow0$ as $t\rightarrow\infty$.
	Note that $e_{\tilde{R}}\rightarrow0$ implies that $\tilde{R}\rightarrow\mathbf{I}_{3}$.
	By \eqref{eq:SLAM_lemm1_2} in Lemma \ref{Lemm:SLAM_Lemma1}, $e_{\tilde{R}}\rightarrow0$
	strictly indicates that $\boldsymbol{\Upsilon}(\tilde{R}M)\rightarrow\underline{\mathbf{0}}_{3}$.
	Thus, from \eqref{eq:SLAM_Filter}, $W_{U}\rightarrow0$ as $E_{i}\rightarrow\underline{\mathbf{0}}_{3}$
	and $\boldsymbol{\Upsilon}(\tilde{R}M)\rightarrow\underline{\mathbf{0}}_{3}$.
	Recall $\tilde{b}_{U}$ and $\dot{\hat{b}}_{U}$ in \eqref{eq:SLAM_b_error}
	and \eqref{eq:SLAM_Filter}. $\dot{\tilde{b}}_{U}=-\dot{\hat{b}}_{U}$
	implies that $\dot{\tilde{b}}_{U}\rightarrow0$ as $E_{i}\rightarrow\underline{\mathbf{0}}_{3}$
	and $\boldsymbol{\Upsilon}(\tilde{R}M)\rightarrow\underline{\mathbf{0}}_{3}$.
	As such, $\tilde{b}_{U}$ is bounded for all $t\geq0$. In addition,
	from \eqref{eq:SLAM_Filter}, $\dot{{\rm \hat{p}}}_{i}\rightarrow0$
	as $E_{i}\rightarrow0$ and $W_{U}\rightarrow0$. In accordance with
	the aforementioned discussion and based on the fact that $\lim_{t\rightarrow\infty}\dot{e}_{i}=0$,
	one obtains
	\[
	\lim_{t\rightarrow\infty}\dot{e}_{i}=\lim_{t\rightarrow\infty}-[\begin{array}{cc}
	-\hat{R}\left[y_{i}\right]_{\times} & \hat{R}\end{array}]\tilde{b}_{U}=0
	\]
\begin{algorithm}
	\caption{\label{alg:Alg1}Systematic convergence nonlinear observer for SLAM}
	
	\textbf{Initialization}:
	\begin{enumerate}
		\item[{\footnotesize{}1:}] Set $\hat{R}\left(0\right)\in\mathbb{SO}\left(3\right)$ and $\hat{P}\left(0\right)\in\mathbb{R}^{3}$.
		As an alternative, construct $\hat{R}\left(0\right)\in\mathbb{SO}\left(3\right)$
		using one of the attitude determination methods in \cite{hashim2020AtiitudeSurvey}\vspace{1mm}
		\item[{\footnotesize{}2:}] Set ${\rm \hat{p}}_{i}\left(0\right)\in\mathbb{R}^{3}$ for all $i=1,2,\ldots,n$\vspace{1mm}
		\item[{\footnotesize{}3:}] Set $\hat{b}_{\Omega}\left(0\right)=\hat{b}_{V}\left(0\right)=0_{3\times1}$
		\vspace{1mm}
		\item[{\footnotesize{}4:}] Select $\xi_{\star}^{0}$, $\xi_{\star}^{\infty}$, $\ell_{\star}$,
		$\bar{\delta}_{\star}=\underline{\delta}_{\star}$, $k_{w}$, $k_{1}$,
		$k_{2}$, $\Gamma$, and $\alpha_{i}$ as positive constants
	\end{enumerate}
	\textbf{while$\hspace{6em}$$\forall i=1,2,\ldots,n,\text{ and }j=1,2,\ldots n_{{\rm R}}$}
	\begin{enumerate}
		\item[{\footnotesize{}5:}] Recall measurements and observations in \eqref{eq:SLAM_Vect_R} and
		their normalization $\upsilon_{j}^{r}=\frac{r_{j}}{\left\Vert r_{j}\right\Vert },\upsilon_{j}^{a}=\frac{a_{j}}{\left\Vert a_{j}\right\Vert }$
		in \eqref{eq:SLAM_Vector_norm}\vspace{1mm}
		\item[{\footnotesize{}6:}] $\hat{\upsilon}_{j}^{a}=\hat{R}^{\top}\upsilon_{j}^{r}$ as in \eqref{eq:SLAM_vect_R_estimate}\vspace{1mm}
		\item[{\footnotesize{}7:}] $M=\sum_{j=1}^{n_{{\rm R}}}s_{j}\upsilon_{j}^{r}\left(\upsilon_{j}^{r}\right)^{\top}$
		as in \eqref{eq:SLAM_M} with $\breve{\mathbf{M}}={\rm Tr}\left\{ M\right\} \mathbf{I}_{3}-M$\vspace{1mm}
		\item[{\footnotesize{}8:}] $\boldsymbol{\Upsilon}=\hat{R}\sum_{j=1}^{n_{{\rm R}}}\left(\frac{s_{j}}{2}\hat{\upsilon}_{j}^{a}\times\upsilon_{j}^{a}\right)$
		as in \eqref{eq:SLAM_VEX_VM}\vspace{1mm}
		\item[{\footnotesize{}9:}] $\pi={\rm Tr}\left\{ \left(\sum_{j=1}^{n_{{\rm R}}}s_{j}\upsilon_{j}^{a}\left(\upsilon_{j}^{r}\right)^{\top}\right)\left(\sum_{j=1}^{n_{{\rm R}}}s_{j}\hat{\upsilon}_{j}^{a}\left(\upsilon_{j}^{r}\right)^{\top}\right)^{-1}\right\} $
		as in \eqref{eq:SLAM_Gamma_VM}\vspace{1mm}
		\item[{\footnotesize{}10:}] $e_{i}=\hat{{\rm p}}_{i}-\hat{R}y_{i}-\hat{P}$ as in \eqref{eq:SLAM_e_Final}\vspace{1mm}
		\item[{\footnotesize{}11:}] $\xi_{\star}=\left(\xi_{\star}^{0}-\xi_{\star}^{\infty}\right)\exp(-\ell_{\star}t)+\xi_{\star}^{\infty}$
		as in \eqref{eq:SLAM_Presc} and $\Lambda_{\star}$ as in \eqref{eq:SLAM_Aux1}\vspace{1mm}
		\item[{\footnotesize{}12:}] $E_{\star}=\frac{1}{2}\text{ln}\frac{\bar{\delta}_{\star}+e_{\star}/\xi_{\star}}{\underline{\delta}_{\star}-e_{\star}/\xi_{\star}}$
		as in \eqref{eq:SLAM_trans2}\vspace{1mm}
		\item[{\footnotesize{}13:}] $W_{\Omega}=\frac{k_{w}\Lambda_{\tilde{R}}-4\mu_{\tilde{R}}}{\tau_{R}}\hat{R}^{\top}\boldsymbol{\Upsilon}$,
		with $\tau_{R}=\underline{\lambda}(\breve{\mathbf{M}})\times(1+\pi[k])$\vspace{1mm}
		\item[{\footnotesize{}14:}] $W_{V}=-\sum_{i=1}^{n}\frac{k_{2}}{\alpha_{i}}\hat{R}^{\top}\Lambda_{i}E_{i}$\vspace{1mm}
		\item[{\footnotesize{}15:}] $\dot{\hat{R}}=\hat{R}\left[\Omega_{m}-\hat{b}_{\Omega}-W_{\Omega}\right]_{\times}$\vspace{1mm}
		\item[{\footnotesize{}16:}] $\dot{\hat{P}}=\hat{R}(V_{m}-\hat{b}_{V}-W_{V})$\vspace{1mm}
		\item[{\footnotesize{}17:}] $\dot{\hat{b}}_{\Omega}=\frac{\Gamma_{1}}{2}\hat{R}^{\top}\Lambda_{\tilde{R}}\boldsymbol{\Upsilon}-\sum_{i=1}^{n}\frac{\Gamma_{2}}{\alpha_{i}}\left[y_{i}\right]_{\times}\hat{R}^{\top}\Lambda_{i}E_{i}$\vspace{1mm}
		\item[{\footnotesize{}18:}] $\dot{\hat{b}}_{V}=-\sum_{i=1}^{n}\frac{\Gamma_{2}}{\alpha_{i}}\hat{R}^{\top}\Lambda_{i}E_{i}$\vspace{1mm}
		\item[{\footnotesize{}19:}] $\dot{{\rm \hat{p}}}_{i}=-k_{1}\Lambda_{i}E_{i}+\hat{R}\left[y_{i}\right]_{\times}W_{\Omega}$
	\end{enumerate}
	\textbf{end while}
\end{algorithm}	
	Consistently with Assumption \ref{Assumption:Feature}, define 
	\[
	Q=-\left[\begin{array}{cc}
	-\hat{R}[y_{1}]_{\times} & \hat{R}\\
	\vdots & \vdots\\
	-\hat{R}[y_{n}]_{\times} & \hat{R}
	\end{array}\right]\in\mathbb{R}^{3n\times6},\hspace{1em}n\geq3
	\]
	Hence, $Q$ is full column rank such that $\lim_{t\rightarrow\infty}Q\tilde{b}_{U}=0$
	results in $\lim_{t\rightarrow\infty}\tilde{b}_{U}=0$. Therefore,
	$\ddot{V}$ is bounded. On the grounds of Barbalat Lemma, $\dot{V}$
	is uniformly continuous. According to the fact that $\tilde{b}_{U}\rightarrow0$
	and $W_{U}\rightarrow0$ as $t\rightarrow\infty$, $\dot{\tilde{\boldsymbol{T}}}\rightarrow0$,
	and in turn, $\tilde{\boldsymbol{T}}\rightarrow\boldsymbol{T}_{c}(\mathbf{I}_{3},P_{c})$,
	where $\boldsymbol{T}_{c}(\mathbf{I}_{3},P_{c})\in\mathbb{SE}\left(3\right)$
	is a constant matrix with $P_{c}\in\mathbb{R}^{3}$ describing a constant
	vector. Consequently, $\lim_{t\rightarrow\infty}\tilde{P}=P_{c}$
	completing the proof.\end{proof}

The complete implementation steps are described in Algorithm \ref{alg:Alg1}.

\begin{figure}[h]
	\centering{}\includegraphics[scale=0.28]{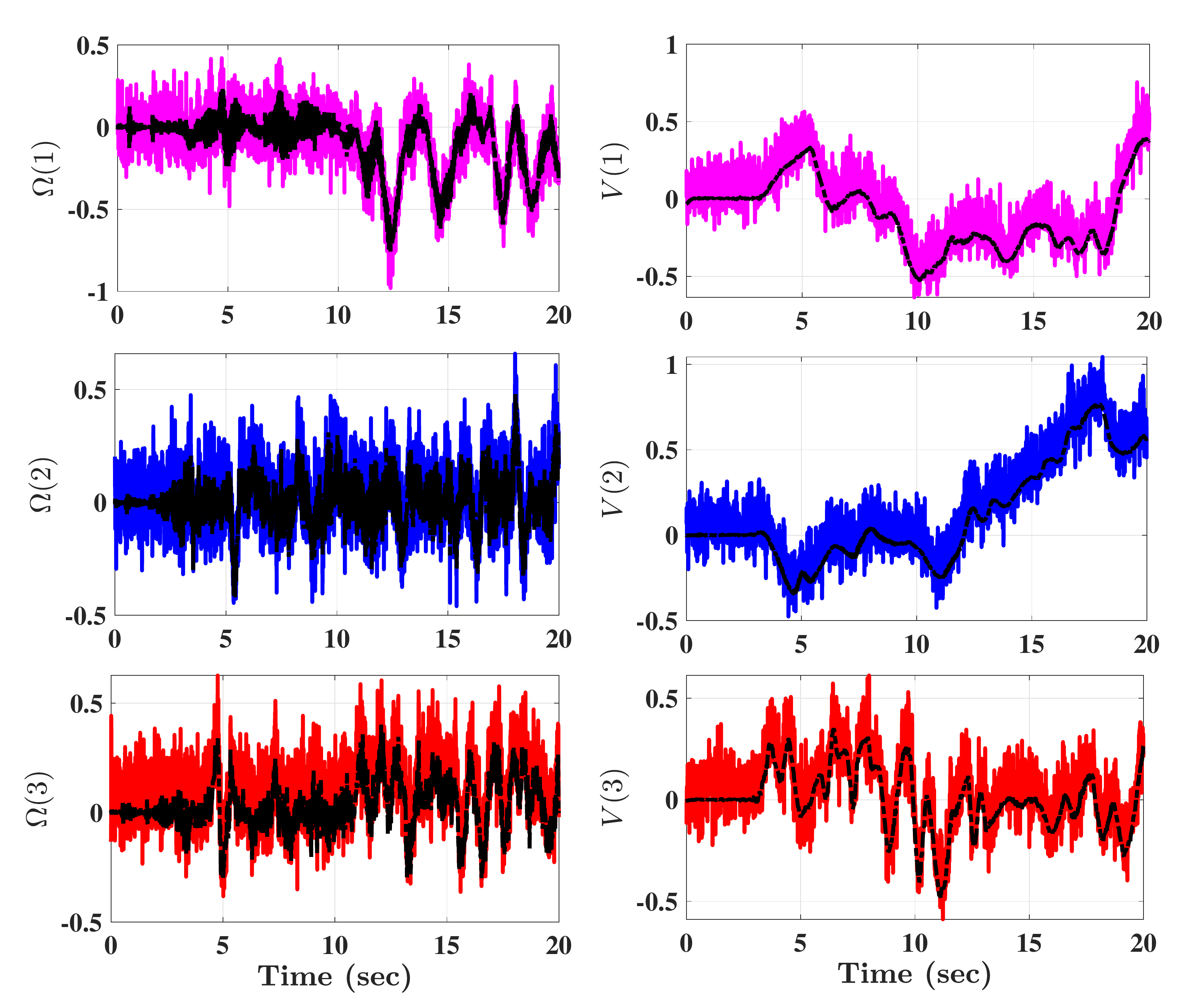}\caption{Angular and translational velocities: True vs measurements with extra
		additive noise.}
	\label{fig:SLAM_Vel}
\end{figure}

\section{Experimental Results \label{sec:SE3_Simulations}}

In this section, the performance of the proposed nonlinear SLAM observer
developed on the Lie group $\mathbb{SLAM}_{n}\left(3\right)$ and
characterized by systematic convergence is tested and evaluated. The
observing capabilities and robustness of the proposed approach against
noise are demonstrated using a real-world EuRoc dataset \cite{Burri2016Euroc}.
The dataset includes 1) ground truth data that consists of the unmanned
aerial vehicle true flight trajectory, 2) IMU measurements, and 3)
stereo images. Due to the lack of real-world landmarks in the dataset,
a set of virtual landmarks defined in the previous section is utilized.
According to the dataset, the true attitude and position of the robot
in 3D space are initialized at
\begin{align*}
P\left(0\right) & =[-1.0761,0.4925,1.3299]^{\top},\\
R\left(0\right) & =\left[\begin{array}{ccc}
-0.2645 & -0.0014 & -0.9644\\
0.0198 & 0.9998 & -0.0069\\
0.9642 & -0.0210 & -0.2645
\end{array}\right]\in\mathbb{SO}\left(3\right)
\end{align*}
Consider four inertial-frame-fixed landmarks placed at ${\rm p}_{1}=[2,0,0]^{\top}$,
${\rm p}_{2}=[-2,0,0]^{\top}$, ${\rm p}_{3}=[0,2,0]^{\top}$, and
${\rm p}_{4}=[0,-2,0]^{\top}$. The initial estimates of the robot
attitude and position in 3D are as follows: 
\begin{align*}
\hat{P}\left(0\right) & =[0,0,0]^{\top},\\
\hat{R}\left(0\right) & =\left[\begin{array}{ccc}
-0.2653 & -0.0642 & -0.9620\\
0.0032 & 0.9977 & -0.0675\\
0.9642 & -0.0210 & -0.2645
\end{array}\right]\in\mathbb{SO}\left(3\right)
\end{align*}
\begin{figure*}
	\centering{}\includegraphics[scale=0.3]{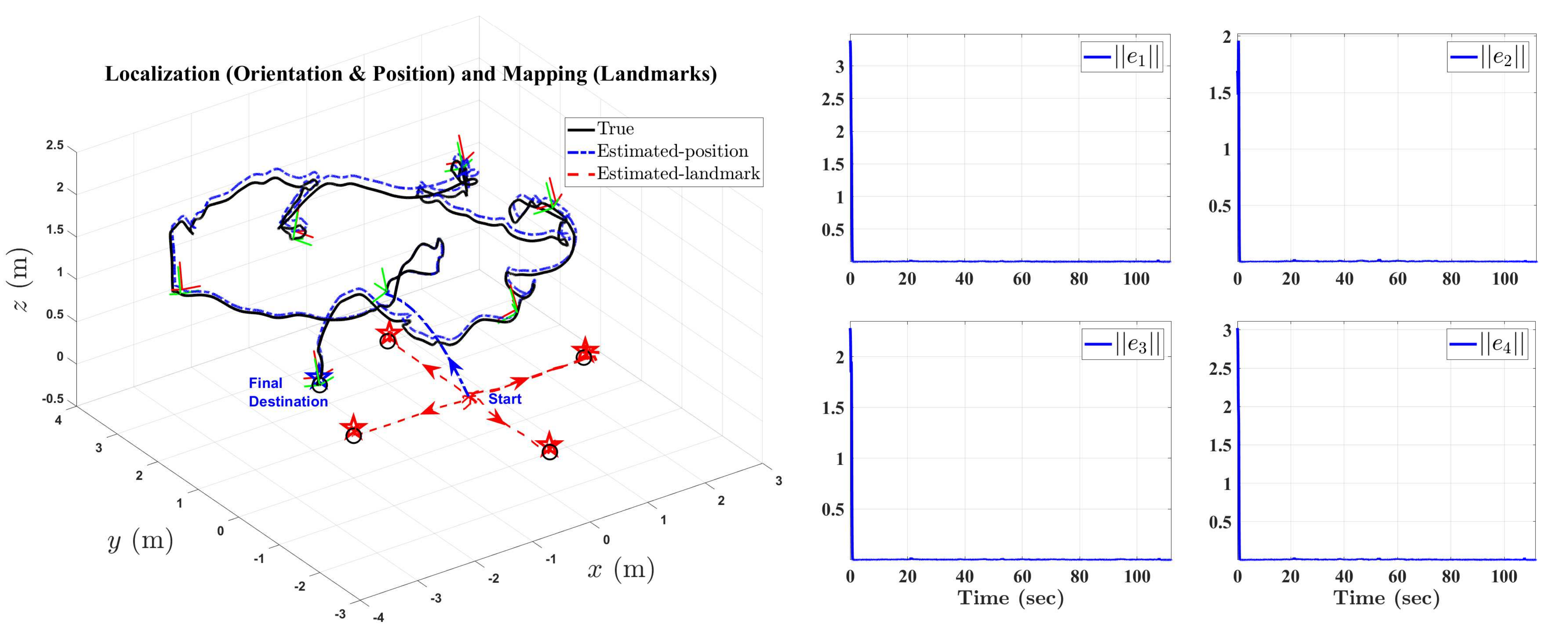}\caption{Experimental validation using dataset Vicon Room 2 01. Output performance
		of the proposed nonlinear observer for SLAM vs the true trajectory
		(on the left), and error trajectories (on the right).}
	\label{fig:SLAM_3d}
\end{figure*}
Also, suppose that the initial estimates of the four landmark positions
are $\hat{{\rm p}}_{1}\left(0\right)=\hat{{\rm p}}_{2}\left(0\right)=\hat{{\rm p}}_{3}\left(0\right)=\hat{{\rm p}}_{4}\left(0\right)=[0,0,0]^{\top}$.
Let us supplement the dataset measurement noise with additional uncertainties.
Thus, assume the group velocity vector $U_{m}$ to be corrupted with
constant bias $b_{U}=\left[b_{\Omega}^{\top},b_{V}^{\top}\right]^{\top}$
such that $b_{\Omega}=[-0.0023,0.0249,0.0816]^{\top}({\rm rad/sec})$
and $b_{V}=[-0.0209,0.1216,0.0788]^{\top}({\rm m/sec})$. Moreover,
let $U_{m}$ be contaminated with noise $n_{U}=\left[n_{\Omega}^{\top},n_{V}^{\top}\right]^{\top}$
such that $n_{\Omega}=\mathcal{N}\left(0,0.1\right)({\rm rad/sec})$
and $n_{V}=\mathcal{N}\left(0,0.1\right)({\rm m/sec})$ where $n_{\Omega}=\mathcal{N}\left(0,0.1\right)$
abbreviates randomly distributed noise with a zero mean and a standard
deviation of $0.1$. The design parameters are selected as $\xi_{\star}^{\infty}=0.03$,
$\ell_{\star}=1$, $\xi_{\star}^{0}=\bar{\delta}_{\star}=\underline{\delta}_{\star}=e_{\star}\left(0\right)+4$,
$\alpha_{i}=0.05$, $\Gamma_{1}=3\mathbf{I}_{3}$, $\Gamma_{2}=10\mathbf{I}_{3}$,
$k_{w}=5$, $k_{1}=10$, and $k_{2}=10$ for all $i=1,2,3,4$. Additionally,
set the initial bias estimate to $\hat{b}_{U}\left(0\right)=\underline{\mathbf{0}}_{6}$.

Fig. \ref{fig:SLAM_Vel} depicts the true angular and translational
velocities plotted against their measured values heavily contaminated
with noise and bias. Black center-line represents the true values,
while magenta, blue, and red stand for the measured values. The output
trajectories estimated by the proposed nonlinear observer for SLAM
are presented in Fig. \ref{fig:SLAM_3d} where they are compared to
the true trajectories of travel. In the left portion of Fig. \ref{fig:SLAM_3d},
the robot trajectory of travel is indicated using a black center-line,
while the final destinations of both the robot and the fixed landmarks
are marked with black circles. The true vehicle orientation is shown
with a green solid-line. The estimated trajectory of the robot's motion
is depicted as a blue center-line that starts at the origin and arrives
to the final destination marked with a blue star $\textcolor{blue}{\star}$. The estimated
landmark trajectories are depicted as red dash-lines and the final
landmark positions estimated by the observer appear as red stars $\textcolor{red}{\star}$.
The vehicle orientation estimation is presented as a red solid-line.
In the right portion of Fig. \ref{fig:SLAM_3d}, the error convergence
of $||e_{1}||$, $||e_{2}||$, $||e_{3}||$, and $||e_{4}||$ is plotted
as a blue solid-line.

It becomes evident that, as illustrated by the left portion of Fig.
\ref{fig:SLAM_3d}, both robot's and landmark position estimates,
initiated at the origin, successfully arrive to their true final destinations.
The right portion of Fig. \ref{fig:SLAM_3d} illustrates asymptotic
convergence of the error trajectories of $e_{i}$ for the nonlinear
observer for SLAM. As attested by Fig. \ref{fig:SLAM_3d}, the proposed
observer is able to successfully localize the unknown robot's position
while simultaneously mapping the unknown environment.%
\begin{comment}
\begin{figure}[h]
\centering{}\includegraphics[scale=0.3]{Fig_e_PPF.eps}\caption{Convergence of the error trajectories of $e_{i}=[e_{i1},e_{i2},e_{i3}]^{\top}$for
$i=1,2,3,4$ depicted in colors with respect to dynamically reducing
boundaries depicted in black.}
\label{fig:SLAM_error1}
\end{figure}
\end{comment}

\begin{comment}
\begin{figure}[h]
\centering{}\includegraphics[scale=0.3]{Fig_e.eps}\caption{Error trajectories of $||\hat{R}R^{\top}||_{{\rm I}}$ and $||{\rm p}_{i}-{\rm \hat{p}}_{i}||$
for all $i=1,2,3,4$.}
\label{fig:SLAM_error2}
\end{figure}
\end{comment}

\section{Conclusion \label{sec:SE3_Conclusion}}

In this paper, a nonlinear observer on the Lie group of $\mathbb{SLAM}_{n}\left(3\right)$
for Simultaneous Localization and Mapping (SLAM) has been proposed.
The observer is able to efficiently control the error by implementing
predefined measures of transient and steady-state performance. For
successful operation the observer requires measurements of angular
and translational velocity, landmark measurements, and an inertial
measurement unit (IMU). Systematic convergence is enabled by the transformation
of the constrained error to its unconstrained form. The proposed observer
design successfully handles unknown bias inevitably present in the
velocity measurements. Experimental results reveal the robustness
and effectiveness of the proposed nonlinear observer for concurrent
estimation of the unknown pose and mapping of the unknown landmarks.

\section*{Acknowledgment}

The authors would like to thank \textbf{Maria Shaposhnikova} for proofreading
the article.

\bibliographystyle{IEEEtran}
\bibliography{bib_SLAM}

\section*{AUTHOR INFORMATION}
\vspace{10pt}

{\bf Hashim A. Hashim} (Member, IEEE) is an Assistant Professor with the Department of Engineering and Applied Science, Thompson Rivers University, Kamloops, British Columbia, Canada. He received the B.Sc. degree in Mechatronics, Department of Mechanical Engineering from Helwan University, Cairo, Egypt, the M.Sc. in Systems and Control Engineering, Department of Systems Engineering from King Fahd University of Petroleum \& Minerals, Dhahran, Saudi Arabia, and the Ph.D. in Robotics and Control, Department of Electrical and Computer Engineering at Western University, Ontario, Canada.\\
His current research interests include stochastic and deterministic attitude and pose filters, Guidance, navigation and control, simultaneous localization and mapping, control of multi-agent systems, and optimization techniques.

\underline{Contact Information}: \href{mailto:hhashim@tru.ca}{hhashim@tru.ca}.
\vspace{50pt}

{\bf Abdelrahman E.E. Eltoukhy} received his BSc Degree in Production Engineering from Helwan University, Egypt, and obtained his MSc in Engineering and Management from the Politecnico Di Torino, Italy. He obtained his PhD degree from The Hong Kong Polytechnic University, Hong Kong. He is currently a Research Assistant Professor in Industrial and Systems Engineering department, The Hong Kong Polytechnic University, Hong Kong.\\
His current research interests include airline schedule planning, logistics and supply chain management, operations research, and simulation.

\end{document}